\documentclass[twocolumn]{IEEEtran}

\usepackage{amsmath,epsfig,amssymb,verbatim,amsopn,cite,multirow}
\usepackage{amsthm}
\usepackage{balance}
\usepackage{multirow}
\usepackage[usenames,dvipsnames]{color}
\usepackage[all]{xy}  %%% used to make block diagram
\usepackage{url}
\usepackage{amsfonts}
\usepackage{amssymb}
\usepackage{epsfig}
\usepackage{epstopdf}
\usepackage{bm}
\usepackage{balance}
\usepackage{graphicx}
\usepackage{subcaption}
\usepackage{footnote}
\usepackage{booktabs}
\usepackage{algorithm,algorithmic}
\usepackage{makecell}

\usepackage[nodisplayskipstretch]{setspace}
\newtheorem{Theorem}{Theorem}

  {\proof}{\proofend}

\usepackage{tikz}
\newcommand{\tikzxmark}{%
\tikz[scale=0.15] {
    \draw[line width=0.6,line cap=round] (0,0) to [bend left=5] (1,1);
    \draw[line width=0.6,line cap=round] (0.1,0.85) to [bend right=3] (0.8,0.05);
}}

\newcommand{\qa}{{\bf a}}

\newcommand{\qe}{{\bf e}}

\newcommand{\qg}{{\bf g}}

\newcommand{\qt}{{\bf t}}

\newcommand{\qv}{{\bf v}}
\newcommand{\qw}{{\bf w}}
\newcommand{\qx}{{\bf x}}

\newcommand{\qB}{{\bf B}}

\newcommand{\qG}{{\bf G}}

\newcommand{\qI}{{\bf I}}

\newcommand{\ettall}{\emph{et al.}}

\newcommand{\dl}{\mathtt{dl}}

\newcommand{\UE}{\mathtt{I}}
\newcommand{\sn}{\mathtt{E}}
\DeclareMathOperator{\ETA}{\boldsymbol{\eta}}
\DeclareMathOperator{\ETAI}{\boldsymbol{\eta}^{\mathtt{I}}}
\DeclareMathOperator{\ETAE}{\boldsymbol{\eta}^{\mathtt{E}}}
\DeclareMathOperator{\aaa}{\mathbf{a}}

\DeclareMathOperator{\MM}{\mathcal{M}}

\DeclareMathOperator{\K}{\mathcal{K}}

\DeclareMathOperator{\C}{\mathbb{C}}

\newcommand{\Ntx}{N}
\newcommand{\Nrx}{N}

\DeclareMathOperator{\Kd}{\mathcal{K}}
\DeclareMathOperator{\LL}{\mathcal{L}}

\newcommand{\NL}{\mathrm{NL}}
\newcommand{\PZF}{\mathrm{PZF}}
\newcommand{\PMRT}{\mathrm{PMRT}}

\newcommand{\EE}{\mathrm{EE}}

\newcommand{\wimk}{\qw_{\mathrm{I},mk}}
\newcommand{\wimkp}{\qw_{\mathrm{I},mk'}}

\newcommand{\weml}{\qw_{\mathrm{E},m\ell}}
\newcommand{\wemlp}{\qw_{\mathrm{E},m\ell'}}
\newcommand{\Ghms}{\hat{\qG}_m^{\sn}}
\newcommand{\Ghmu}{\hat{\qG}_m^{\UE}}
\newcommand{\Snn}{\sigma_n^2}
\newcommand{\Ex}{\mathbb{E}}
\newcommand{\yel}{y_{\mathtt{E},\ell}}
\newcommand{\nel}{n_{\mathtt{E},\ell}}
\newcommand{\yik}{y_{\mathtt{I},k}}
\newcommand{\nik}{n_{\mathtt{I},k}}

\newcommand{\bsHz}{\text{[bit/s/Hz]}}

\newcommand{\HEQoS}{{\Xi}_{\ell}}
\newcommand{\SEQoS}{\mathcal{S}_\dl^o}
\newcommand{\tilSEQoS}{\tilde{\mathcal{S}}_\dl^o}
\newcommand{\barSEQoS}{\bar{\mathcal{S}}_\dl^o}

\newcommand{\tgmkue}{\tilde{\qg}_{mk}^{\UE}}
\newcommand{\tgmls}{\tilde{\qg}_{m\ell}^{\sn}}

\newcommand{\gmkue}{\qg_{mk}^{\UE}}

\newcommand{\hgmkue}{\hat{\qg}_{mk}^{\UE}}

\newcommand{\hgmls}{\hat{\qg}_{m\ell}^{\sn}}

\newcommand{\gmls}{\qg_{m\ell}^{\sn}}

\newcommand{\gtilmls}{\tilde{\qg}_{m\ell}^{\sn}}

\newcommand{\ghmons}{\hat{\qg}_{m1}^{\sn}}
\newcommand{\ghmLs}{\hat{\qg}_{mL}^{\sn}}
\newcommand{\ghmonue}{\hat{\qg}_{m1}^{\UE}}
\newcommand{\ghmKdue}{\hat{\qg}_{mK_d}^{\UE}}
\newcommand{\ghmkue}{\hat{\qg}_{mk}^{\UE}}
\newcommand{\gtilmkue}{\tilde{\qg}_{mk}^{\UE}}

\newcommand{\gamuemk}{\gamma_{mk}^{\UE}}

\newcommand{\gamsml}{\gamma_{m\ell}^{\sn}}

\newcommand{\betamkue}{\beta_{mk}^{\UE}}

\newcommand{\betamls}{\beta_{m\ell}^{\sn}}

\newcommand{\etamI}{\eta_{m}^{\mathtt{I}}}
\newcommand{\etamE}{\eta_{m}^{\mathtt{E}}}

\newcommand{\etamkI}{\eta_{mk}^{\mathtt{I}}}

\newcommand{\etamkIn}{\eta_{mk}^{{\mathtt{I}}^{(n)}}}
\newcommand{\etamkpI}{\eta_{mk'}^{\mathtt{I}}}
\newcommand{\etamlE}{\eta_{m\ell}^{\mathtt{E}}}

\newcommand{\etamlEn}{\eta_{m\ell}^{{\mathtt{E}}^{(n)}}}
\newcommand{\etamlpE}{\eta_{m\ell'}^{\mathtt{E}}}

\newcommand{\etamIn}{\eta_{m}^{\mathtt{I}(n)}}
\newcommand{\etamEn}{\eta_{m}^{\mathtt{E}(n)}}

\newcommand{\umln}{u_{m\ell}^{(n)}}
\newcommand{\uml}{u_{m\ell}}

\newcommand{\Stre}{{\mathtt{S}_3}}
\newcommand{\Sto}{{\mathtt{S}_2}}

\newcommand{\SEk}{\mathrm{SE}_k}
\newcommand{\SINRk}{\mathrm{SINR}_k}

\newcommand{\xik}{x_{\mathtt{I},k}}
\newcommand{\xikp}{x_{\mathtt{I},k'}}
\newcommand{\xel}{x_{\mathtt{E},\ell}}
\newcommand{\xelp}{x_{\mathtt{E},\ell'}}

\newcommand{\tilgmal}{\tilde{\Gamma}_{\ell}}

\newcommand{\rodsnz}{\frac{\rho \Sn}{\zeta_m}}

\newcommand{\Pbhm}{P_{\mathtt{bh},m}}
\newcommand{\Sn}{\sigma_n^2}

\newtheorem{remark}{Remark}

\title{Cell-Free Massive MIMO-Assisted SWIPT for IoT Networks}
\author{Mohammadali Mohammadi,~\IEEEmembership{Senior Member,~IEEE,} 
Le-Nam Tran,~\IEEEmembership{Senior Member,~IEEE,}\\
Zahra  Mobini,~\IEEEmembership{Member,~IEEE,} 
Hien Quoc Ngo,~\IEEEmembership{Fellow,~IEEE,}   and  Michail Matthaiou,~\IEEEmembership{Fellow,~IEEE}
\thanks{This work was supported in part by the U.K. Engineering and Physical Sciences Research Council (EPSRC) (grant No. EP/X04047X/1). The work of  H. Q. Ngo and Z.  Mobini
 was supported by the U.K. Research and Innovation Future
Leaders Fellowships under Grant MR/X010635/1, and a research grant from the Department for the Economy Northern Ireland under the US-Ireland R\&D Partnership Programme. The work of M. Mohammadi and M. Matthaiou was supported by the European
Research Council (ERC) under the European Union’s Horizon 2020 research
and innovation programme (grant agreement No. 101001331). The work of L.-N. Tran was supported in part by Taighde Eireann - Research Ireland under Grant numbers 22/US/3847 and 13/RC/2077\_P2.}
\thanks{Mohammadali Mohammadi, Zahra Mobini, Hien Quoc Ngo, and Michail Matthaiou are with the Centre for Wireless Innovation (CWI), Queen's University Belfast, BT3 9DT Belfast, U.K.,
(email:\{m.mohammadi, zahra.mobini, hien.ngo, m.matthaiou\}@qub.ac.uk).}
\thanks{Le-Nam Tran is with the School of Electrical and Electronic Engineering, University College Dublin, Dublin, D04 V1W8, Ireland, (e-mail:
nam.tran@ucd.ie).}
\thanks{ Parts of this paper appeared at the 2023 IEEE GLOBECOM conference~\cite{Mohammadi:GC:2023}.
}}

\allowdisplaybreaks

\begin{document}

\bstctlcite{IEEEexample:BSTcontrol}
\maketitle

\begin{abstract} 
This paper studies cell-free massive multiple-input multiple-output (CF-mMIMO) systems that underpin  simultaneous wireless information and power transfer (SWIPT) for separate information users (IUs) and energy users (EUs) in Internet of Things (IoT) networks. We propose a joint access point (AP) operation mode selection and power control design, wherein certain APs are designated  for energy transmission to EUs, while others  are dedicated to information transmission to IUs.  The performance of the system, from both a spectral efficiency (SE) and energy efficiency (EE) perspective, is comprehensively analyzed. Specifically, we formulate two mixed-integer nonconvex optimization problems for maximizing the average sum-SE and EE, under realistic power consumption models and constraints on the minimum individual SE requirements for individual IUs, minimum HE for individual EUs, and maximum transmit power at each AP.  The challenging optimization problems are solved using successive convex approximation (SCA) techniques. The proposed  framework design is further applied to the average sum-HE maximization and energy harvesting fairness problems. Our numerical results demonstrate that the proposed joint AP operation mode selection and power control algorithm can achieve EE performance gains of up to $4$-fold and $5$-fold over random AP operation mode selection, with and without power control respectively.

%195 words
\end{abstract}


%==============================================================================
%==============================================================================
\section{Introduction}
%==============================================================================
%==============================================================================
Cell-free massive multiple-input multiple-output (CF-mMIMO) stands out as a leading candidate among the potential technologies for the upcoming sixth-generation (6G) wireless networks~\cite{Matthaiou:COMMag:2021,Ngo:PROC:2024}. CF-mMIMO represents a major leap from fifth-generation (5G) cellular networks, which relied on massive MIMO base stations (BSs). This advancement not only circumvents the inter/intra-cell interference and the cell-edge problems, both major challenges within cellular networks, but also promises to meet the future requirements of ubiquitous connectivity and addressing the ever-growing demands of data traffic. In CF-mMIMO, each user equipment (UE) is coherently served by a large number of distributed access points (APs), which are connected to the central processing unit (CPU) through fronthaul links. As a result of this user-centric paradigm, cell boundaries vanish, while the UEs enjoy seamless and uniform coverage and service~\cite{Ngo:PROC:2024,Hien:cellfree}. Moreover, CF-mMIMO avails of both macro-diversity gain and low path loss, since the service antennas are closer to the UEs. Therefore, it can enable spectral efficiency (SE) and energy efficiency (EE) enhancements by orders of magnitude compared to conventional cellular massive MIMO systems~\cite{ngo18TGN}.  

Considering the EE perspective, as one of the main performance metrics of future wireless networks, activating all APs solely for communication within the network may not be necessary, particularly when the SE requirements of the UEs can be fulfilled using a smaller subset of APs~\cite{ngo18TGN}.  This underscores the potential of CF-mMIMO to facilitate the distributed implementation of dual-functional applications, such as integrated sensing and communication (ISAC), simultaneous wireless information and power transfer (SWIPT), as well as applications related to wireless physical layer security. In this regard, a subset of APs can be designated for communication purposes as the primary function, while the remaining APs can be allocated for the secondary applications. However, the assignment of operation modes at the APs must be based on the specific requirements of each application, as well as the network requirements. More importantly, leveraging a smaller distance between the APs and UEs can significantly improve the inter-system interference management in the network compared to  co-located massive MIMO arrays. The distributed implementation of ISAC within CF-mMIMO networks has recently explored in~\cite{elfiatoure2024multiple}. Additionally, in~\cite{Mobini:GC:2023}, the utilization of distributed APs in CF-mMIMO networks has been investigated for implementing a distributed wireless surveillance system. However, the potential for a distributed implementation of SWIPT through CF-mMIMO networks has yet to be fully explored, and this gap in knowledge serves as the primary motivation for our paper.

%==============================================================================
\vspace{-1em}
\subsection{Review of Related Literature }
%==============================================================================
SWIPT emerges as a promising approach for future energy-hungry Internet-of-Things (IoT) networks, offering the possibility to both connect and power wireless devices through radio-frequency (RF) waves. In this approach, specific devices, termed information UEs (IUs), are designated to receive and transmit information using the harvested energy (HE) collected during the wireless power transfer (WPT) phase. Meanwhile, another set of devices, known as energy UEs (EUs), may exist in the network. These devices are specifically designed to harvest energy from the received RF signals for their signal processing and circuit operation tasks. An example of these EUs is sensors, which are used to sense and/or transmit control commands in the network. Nevertheless, the main challenge of WPT in such networks is the low efficiency due to radio scattering and path loss~\cite{Lu:Tut:2015,Ponnimbaduge:tut:2018}. As effective countermeasures, MIMO, and especially massive MIMO techniques, have been adopted in wireless powered communication networks~\cite{Khan:TWC:2018}, since the RF energy becomes more concentrated and, thus, can be more effectively harvested. Although multiple antennas can improve the efficiency of WPT,  the problem of imbalance in serving the network nodes located at varying distances from the BS still persists~\cite{Nguyen:TCOM:2017,Dong:CLET:2018,Asiedu:IOT:2020,Mohammadi:TCOM:2021}. 
% the performance of cell-boundary IUs/EUs is still poor, which results from the heavy path loss as well as interference between the information and power transfer phases. High levels of path loss for IUs/EUs located far from the BS can lead to significant unfairness within the network. In extreme scenarios, IUs/EUs with poor channel conditions may experience such severe degradation that they are never scheduled. 

To mitigate the significant path loss experienced by cell boundary terminals and ensure fairness in the network, CF-mMIMO emerges as a promising alternative. This approach ensures that each terminal is served by adjacent distributed APs, thereby reducing the distances between the terminals and the served APs. CF-mMIMO with WPT has been considered in several works. Wang~\ettall~\cite{Wang:JIOT:2020} proposed a wirelessly powered cell-free IoT system and minimized the total energy consumption of APs  under given signal-to-interference-plus-noise ratio (SINR) constraints during uplink (UL) data transmission using the HE in the previous downlink (DL) WPT phase. Demir~\ettall~\cite{Demir:TWC:2021} considered DL WPT and UL wireless information transfer (WIT) via the HE in CF-mMIMO networks. They maximized the minimum SE of the UEs' under APs' and UEs'
transmission power constraints. Femenias~\ettall~\cite{Femenias:TCOM:2021} investigated a CF-mMIMO system with separated EUs and IUs and formulated a coupled UL/DL power control algorithm to optimize the minimum of the weighted SINR of EUs. Zhang~\ettall~\cite{Zhang:IoT:2022} conducted a performance analysis of DL CF-mMIMO-based IoT SWIPT networks in terms of
both HE and achievable SE for different precoding techniques. Then, they proposed a max–min power control policy to achieve uniform HE and DL SE across all the EUs. Zhang~\ettall~\cite{Zhang:TWC:2023} investigated SWIPT-enabled CF-mMIMO networks with power splitting (PS) receivers and non-orthogonal multiple access. A machine learning-based approach was designed to maximize the sum rate, subject to quality of service (QoS) requirements at each UE and a power budget constraint at each AP by optimizing the UE clustering, the power control coefficients, and the PS ratios. Galappaththige~\ettall~\cite{Galappaththige:WCL:2024} developed a rate-splitting multiple access-assisted  CF-mMIMO network, where the UEs harvest energy with PS receivers. 

In~\cite{Xinjiang:TWC:2021,Yang:SYSJ:2022}, the application of SWIPT in network-assisted full-duplex (NAFD) CF-mMIMO networks has been investigated. These works assume fixed mode assignment at the APs to serve both DL and UL IUs, assuming that DL IUs harvest energy from DL APs using PS. These works, utilize the instantaneous channel state information (CSI) for system level designs rather than the statistical CSI. Therefore, all resource allocation designs must be recomputed rapidly once the small-scale fading coefficients have changed. Moreover, these designs must be implemented at the CPUs, thus, the APs have to send all channel estimates to the CPUs which will entail very large overhead and, consequently, low EE performance, especially in CF-mMIMO where the numbers of APs and UEs are usually very large. In our recent work \cite{Hua:WCNC:2024}, we discussed the integration of beyond-diagonal reconfigurable intelligent surfaces into NAFD CF-mMIMO SWIPT systems, considering random AP mode selection.

%%%%%%%%%%%%%%%%%%%%%%%%%%%%%%%%%%%%%%%%%%%
\begin{table*}
	\centering
	\caption{\label{tabel:Survey} Contrasting our contributions to the CF-mMIMO SWIPT literature}
	\vspace{-0.6em}
	\small
\begin{tabular}{|p{3.8cm}|p{1.7cm}|p{0.65cm}|p{0.65cm}|p{0.65cm}|p{0.65cm}
|p{0.65cm}|p{0.65cm}|p{0.65cm}|p{0.65cm}|p{0.65cm}|}
	\hline
        \centering\textbf{Contributions} 
        &\centering \textbf{This paper}
        &\centering\cite{Wang:JIOT:2020} 
        &\centering\cite{Demir:TWC:2021} 
        &\centering\cite{Femenias:TCOM:2021} 
        &\centering\cite{Zhang:IoT:2022}
        &\centering\cite{Zhang:TWC:2023}
        %&\centering\cite{Braga:TVT:2023}
        &\centering\cite{Galappaththige:WCL:2024}
          &\centering\cite{Xinjiang:TWC:2021}
        &\centering\cite{Yang:SYSJ:2022}
        &\cite{Hua:WCNC:2024}
        \cr

        \hline

        Power allocation     
        & \centering\checkmark  
        & \centering\checkmark  
        & \centering\checkmark  
        & \centering\checkmark     
        & \centering\checkmark 
        & \centering\checkmark
        %& \centering\checkmark 
        & \centering\tikzxmark
        & \centering\checkmark 
        & \centering\checkmark
        & \centering\tikzxmark
         \cr
        
        \hline

       NL-EH model         
        &\centering\checkmark 
        &\centering\tikzxmark  
        &\centering\checkmark 
        &\centering\tikzxmark
        & \centering\tikzxmark
        & \centering\checkmark
        %& \centering
        & \centering\checkmark 
        & \centering\tikzxmark
        & \centering\tikzxmark
        & \centering\checkmark
        \cr

           \hline
        Statistical CSI         
        &\centering\checkmark  
        &\centering\checkmark   
        & \centering \checkmark 
        &\centering \checkmark  
        &\centering\checkmark 
        & \centering\checkmark 
        & \centering\tikzxmark
        & \centering\tikzxmark
        & \centering\tikzxmark
        & \centering\checkmark
        \cr

        \hline 

              EE enhancement     
        &\centering\checkmark 
        &\centering\tikzxmark  
        & \centering\tikzxmark    
        &\centering\tikzxmark
        &\centering\tikzxmark
        &\centering\tikzxmark
        %&
        &\centering\tikzxmark
        &\centering\tikzxmark
        &\centering\tikzxmark
        &\centering\tikzxmark
        \cr
        
        \hline

         AP operation mode selection    
         &\centering\checkmark  
        & \centering\tikzxmark
        & \centering\tikzxmark  
        & \centering\tikzxmark     
        & \centering\tikzxmark
        &\centering\tikzxmark
        & \centering\tikzxmark
        %&
        &\centering\tikzxmark 
        &\centering\tikzxmark
        &\centering\tikzxmark
        \cr

        \hline
\end{tabular}
\vspace{-1.2em}
\label{Contribution}
\end{table*}
%%%%%%%%%%%%%%%%%%%%%%%%%%%%%%%%%%%%%%%%%%%%%%%%%%
%=======================================================================
\vspace{-1em}
\subsection{Research Gap and Main Contributions}
%====================================================================
The primary underlying assumption in~\cite{Wang:JIOT:2020,Demir:TWC:2021,Femenias:TCOM:2021,Zhang:IoT:2022} is the application of a time-splitting or PS protocol for energy harvesting (EH), leading to significant SE loss for the IUs, while potentially failing to meet the HE requirements of the EUs. This challenge is amplified when the EUs utilize the HE for UL data transmission, as they often encounter periods of insufficient energy harvest during the EH phase. Furthermore, even with an optimal power control design and by using the PS design, all these designs would still suffer from the  fundamental limitation of simultaneously increasing both the SE and HE for separate EUs and IUs. This is due to the inefficient use of available resources, as DL WPT towards EUs and  DL (UL) WIT towards IUs (APs) occur over orthogonal time slots. A straightforward approach to enhance both the SE and HE is to  deploy a large number of APs, but this option is not energy efficient due to the large fronthaul burden and transmit power requirements~\cite{ngo18TGN}. 

As an alternative, inspired by the NAFD concept for CF-mMIMO networks~\cite{Mohammadi:JSAC:2023},  we propose a novel network architecture that jointly designs the AP operation mode selection and power control strategy to maximize the EE, SE, and HE under the constraints of per-IU SE and per-EU HE. Specifically, relying on the long-term CSI, the APs are divided into information transmission APs (termed as I-APs) and energy transmission APs (termed as E-APs), which  simultaneously serve IUs and EUs over the whole coherence interval. While this new architecture provides EUs with an opportunity to harvest energy from all APs, it also creates increased interference at the IUs due to concurrent E-AP transmissions. To deal with this problem, we apply local protective partial zero-forcing (PPZF) precoding, where  partial zero-forcing (PZF) precoding and protective maximum ratio transmission (PMRT) are used at the I-APs and E-APs, respectively, to guarantee full protection for the IUs against energy signals intended for the EUs. The main contributions of this paper are:
\begin{itemize}
    \item We derive  closed-form expressions for the DL SE of the IUs, the average input energy to the non-linear (NL) EH receiver at EUs, and the overall EE of the network. Then, we formulate sum-SE and EE maximization problems, through the joint AP operation mode selection and power control, considering  the per-AP power constraints, as well as the individual SE and HE constraints for IUs and EUs.   
    \item We develop an iterative algorithm to solve the complicated binary non-convex optimization problems. In particular, we transform the formulated problems into more tractable problems with continuous variables only. Then, we apply successive convex approximation (SCA) to solve the resulting problems.
    \item We further extend our proposed design to maximize the average of the sum-HE at the EUs and to maximize the EH fairness among the EUs, subject to QoS requirements at both the IUs and EUs. Moreover, three benchmarks are discussed, and their corresponding optimization problems are presented and solved. 
    \item Our numerical results demonstrate that the proposed architecture  improves significantly the SE and EE  compared to the benchmark schemes. For specific SE and HE requirements, it boosts the EE by an order of magnitude, compared to  conventional designs via orthogonal transmission through time division between information and energy transfer. Moreover, while conventional designs in the literature fail to support dense networks, our proposed scheme can efficiently accommodate them.
\end{itemize}    

A comparison of our contributions against the state of the art in the space of CF-mMIMO SWIPT is tabulated in Table~\ref{tabel:Survey}.

% %==============================================================================
\vspace{0em}
 \subsection{Paper Organization and Notation}
% %==============================================================================
The rest of this paper is organized as follows: In Section~\ref{sec:Sysmodel}, we describe the system model for the proposed CF-mMIMO SWIPT system and the formulation of sum-SE and EE enhancement problems. The proposed solutions for these optimization problems are discussed in Section~\ref{sec:optimization}. The extension of the network design for realizing other objectives, along with benchmarks, is presented in Section~\ref{sec:extension}.  Finally, the numerical results and some discussions are provided in Section~\ref{sec:num}, followed by the conclusion remarks in Section~\ref{sec:conc}.

\textit{Notation:} We use bold upper case letters to denote matrices, and lower case letters to denote vectors. The superscript $(\cdot)^H$ stands for the conjugate-transpose (Hermitian); $\mathbf{I}_N$ denotes the $N\times N$ identity matrix; $(\cdot)^{-1}$ denotes the matrix inverse. A circularly symmetric complex Gaussian distribution having variance $\sigma^2$ is denoted by $\mathcal{CN}(0,\sigma^2)$. Finally, $\mathbb{E}\{\cdot\}$ denotes the statistical expectation.

%--------------------------------
\vspace{0em}
\section{System model}~\label{sec:Sysmodel}
%--------------------------------
We consider a CF-mMIMO system  under  time-division duplex operation, where $M$ APs serve (in DL)  $K_d$ IUs and $L$ EU nodes with EH capability. Each AP is connected to the CPU via a high-capacity fronthaul link.\footnote{ In general, the system should consist of multiple CPUs connected via fronthaul links. Depending on the specific processing requirements, these CPUs can exchange information such as channel estimates, data, and power control coefficients \cite{Ngo:PROC:2024}.} Each IU and EU is equipped with one single antenna, while each AP is equipped with $N$ antennas. All APs, IUs and EUs are half-duplex devices. For notational simplicity, we define the sets $\Kd\triangleq \{1,\dots,K_d\}$ and  $\mathcal{L}\triangleq\{1,\ldots,L\}$ to collect the indices of the IUs and EUs, respectively. Moreover, the set of all APs is denoted by $\MM \triangleq \{1, \dots, M\}$. As shown in Fig.~\ref{fig:systemmodel},  information and energy transmissions are performed simultaneously and in the same frequency band via AP mode selection approaches, i.e., DL information or energy transmission mode at the APs. IUs receive information from a group of the APs (called I-APs) and, at the same time, EUs harvest energy from the remaining APs (called E-APs). EUs consume the average energy they have collected to transmit pilots and data.  Each coherence block includes two phases: 1) UL training for channel estimation; 2) DL WIT and WPT. We assume a quasi-static channel model, with each channel coherence interval spanning $\tau_c$ time slots. The duration of the training is denoted as $\tau$, while the duration of the DL WIT and WPT is $(\tau_c-\tau)$.

%%%%%%%%%%%%%%%%%%%%%%%%%%%%%%%%%%%%%%
\begin{figure}[t]
\vspace{-0.5em}
	\centering
	\vspace{-0.5em}
	\includegraphics[width=0.48\textwidth]{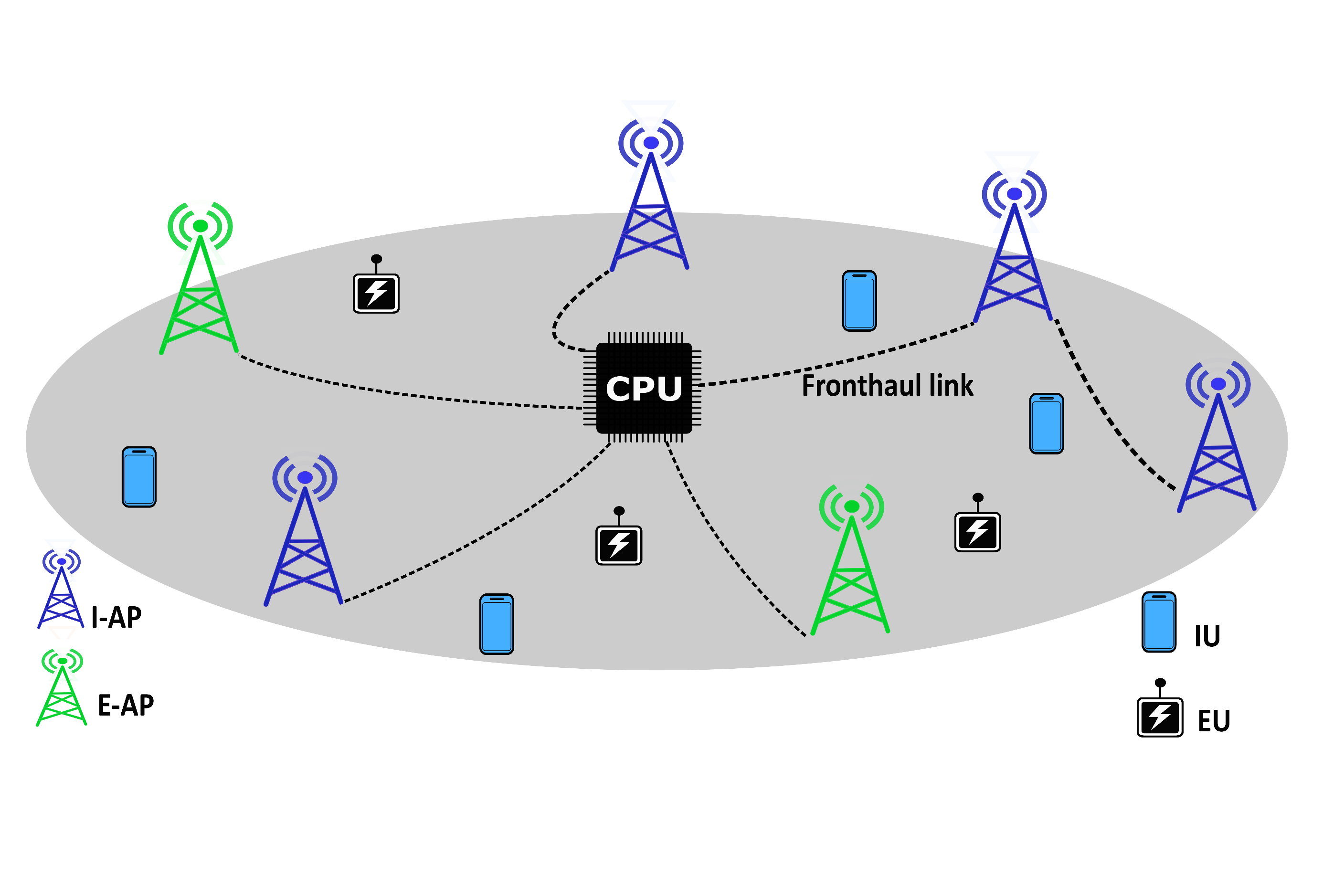}
	 \vspace{-2.7em}
	\caption{The proposed  CF-mMIMO SWIPT system.}
	\vspace{-0.3em}
	\label{fig:systemmodel}
\end{figure}
%%%%%%%%%%%%%%%%%%%%%%%%%%%%%%%%%%%
%--------------------------------
\vspace{-1em}
\subsection{Uplink Training for Channel Estimation}
\label{phase:ULforCE}
%--------------------------------
The channel vector between the $k$-th IU ($\ell$-th EU) and the $m$-th AP is denoted by $\gmkue\in\C^{\Ntx \times 1}$ ($\gmls\in\C^{\Nrx \times 1}$), $\forall m \in \MM$ and $k \in \Kd$ ($\ell \in \LL$). It is modeled as $ \gmkue=\sqrt{\betamkue}\tgmkue,~(\gmls=\sqrt{\betamls}\tgmls) $, where $\betamkue$ ($\betamls$) is the large-scale fading coefficient and $\tgmkue\!\in\!\C^{\Ntx \times 1}$ ($\tgmls\!\in\!\C^{\Ntx \times 1}$) is the small-scale fading vector whose elements are independent and identically distributed $\mathcal{CN} (0, 1)$ random variables.\footnote{ In specific scenarios, such as industrial environments where the propagating wavefronts encounter fewer obstructions from obstacles like tall buildings,  line-of-sight (LoS) propagation tends to dominate~\cite{Zhang:TWC:2024}. Consequently, analyzing the performance of the proposed framework under Ricean fading represents an intriguing avenue for future research. } 

In each coherence block of length $\tau_c$, all IUs and EUs are assumed to transmit their pairwisely orthogonal pilot sequences of length $\tau$ to all the APs, which requires $\tau\geq K_d + L$.\footnote{ The assumption of orthogonal pilots is reasonable in scenarios with a long coherence interval or/and a moderate number of users, as experienced in low or medium mobility environments~\cite{marzetta2016fundamentals}. In contrast, high-mobility or dense networks require pilot reuse among users, demanding advanced analysis and designs, such as pilot contamination cancellation.} At AP $m$, $\gmkue$  and $\gmls$ are estimated by using the received pilot signals and the minimum-mean-square-error (MMSE) estimation technique. 

By following~\cite{Hien:cellfree}, the MMSE estimates $\hgmkue$ and $\hgmls$ of $\gmkue$  and $\gmls$ are $\hgmkue \sim \mathcal{CN}(\boldsymbol{0},\gamuemk \mathbf{I}_N)$, and $\hgmls \sim \mathcal{CN}(\boldsymbol{0},\gamsml \mathbf{I}_N)$, respectively,  where $\gamuemk \triangleq
\frac{{\tau\rho_t}(\betamkue)^2}{\tau\rho_t \betamkue+1}$, and $\gamsml \triangleq 
\frac{{\tau\rho_t}(\betamls)^2}{\tau\rho_t \betamls +1}$, where $\rho_t=\tilde{\rho_t}/\Sn$ is the normalized signal-to-noise ratio (SNR) of
each pilot symbol and $\tilde{\rho_t}$ is the pilot power. The corresponding channel estimation errors are denoted by $\gtilmkue\sim\mathcal{CN}(\boldsymbol{0},(\betamkue-\gamuemk)\qI_N)$ and $\gtilmls\sim\mathcal{CN}(\boldsymbol{0},(\betamls-\gamsml)\qI_N)$.

%=================================================
\vspace{-0.8em}
\subsection{Downlink Information and Power Transmission}
%=================================================
In this phase, the APs are able to select between the information and energy transmission modes and use different precoders among the PZF and PMRT to serve DL IUs and EUs, respectively.  The decision of which mode is assigned to each AP is optimized to achieve the highest sum-SE and total EE of the network, as will be discussed in Section~\ref{sec:optimization}. Note that the AP mode selection is performed on the large-scale fading timescale which changes very slowly with time. The binary variables indicating the mode assignment for each AP $m$, during DL transmission, are defined as
% %---------------------------
\begin{align}
\label{a}
a_{m} &\triangleq
\begin{cases}
  1 & \text{AP $m$ operates in information transmit mode,}\\
  0 & \text{AP $m$ operates in power transmit mode.}
\end{cases}
\end{align}
%-------------

In DL, all I-APs (i.e., APs with $a_m=1$) aim to transmit data symbols $\xik$, with $\Ex\big\{\big\vert\xik\big\vert^2\big\}=1$ to IU $k\in\Kd$. At the same time all E-APs (i.e., APs with $a_m=0$) transmit energy symbol $\xel$, with $\Ex\big\{\big\vert\xel\big\vert^2\big\}=1$ to IU $\ell\in\LL$.  The vector of the transmitted signal from AP $m$ can be expressed as   %------------------------------
\vspace{-0.2em}
\begin{align}~\label{eq:Xm}
   \qx_{m} &= \sqrt{a_m}\underbrace{\sum\nolimits_{k\in\Kd}\sqrt{\rho\etamkI}\wimk \xik}_{\qx_{\mathtt{I},m}}
   \nonumber\\
   &+\sqrt{(1-a_m)}\underbrace{\sum\nolimits_{\ell\in\mathcal{L}}  \sqrt{\rho\etamlE}\weml \xel}_{\qx_ {\mathtt{E},m}},
\end{align}
%------------------------------
where $\rho\triangleq \tilde{\rho}/\Sn $ is the normalized DL SNR, while  $\tilde{\rho}$ is the maximum transmit power of each AP and $\Sn$ is the noise power; $\wimk \in \C^{N\times 1}$ and $\weml\in \C^{N\times 1}$ are the precoding vectors for IU $k$ and EU $\ell$, respectively, with $\Ex\big\{\big\Vert\wimk\big\Vert^2\big\}=1$ and $\Ex\big\{\big\Vert\weml\big\Vert^2\big\}=1$. Note that AP $m$ can only transmit either $\qx_{\mathtt{I},m}$ or $\qx_ {\mathtt{E},m}$, depending on its assigned transmission mode. Moreover, $\etamkI$ and $\etamlE$ are the DL
power control coefficients chosen to satisfy the power constraint at each AP, given by
%--------------------------------
\vspace{0em}
\begin{equation}
a_m\Ex\big\{\big\Vert \qx_{\mathtt{I},m}\big\Vert^2\big\}+(1-a_m)\Ex\big\{\big\Vert \qx_{\mathtt{E},m}\big\Vert^2\big\}\leq \rho.
\end{equation}
%--------------------------------

At IU $k\in\Kd$ and EU $\ell\in\LL$, the received signal can be expressed as
%------------------------------
\vspace{-0.2em}
\begin{align} ~\label{eq:yik}
   \yik &=  \sqrt{\rho}\sum\nolimits_{m\in\MM}
   \sqrt{a_m\etamkI}
   \big(\gmkue\big)^H\wimk \xik
   \\
   &+
   \sqrt{\rho}\sum\nolimits_{m\in\MM}
   \sum\nolimits_{k'\in \Kd\setminus k}
   \sqrt{a_m\etamkpI}
   \big(\gmkue\big)^H\wimkp \xikp  
   \nonumber\\
   &\hspace{0em}
   +\!\sqrt{\rho}\sum_{m\in\MM}\!\sum_{\ell\in\LL}\!
   \sqrt{(1\!-\!a_m)\etamlE}
   \big(\gmkue\big)^{\!H}\weml \xel\! +\! \nik,\nonumber  
\end{align}
%-------------
and
%---------------
\vspace{-0.2em}
\begin{align}
   \yel&= 
   \sqrt{\rho}\sum_{m\in\MM}\sum_{\ell'\in\LL}
   \sqrt{(1-a_m)\etamlpE} \big(\gmls\big)^H\wemlp \xelp \nonumber \\ 
   &+\sqrt{\rho}\!\!
   \sum_{m\in\MM}\sum_{k\in\Kd}\!\!
   \sqrt{a_m\etamkI}\big(\gmls\big)^H\wimk \xik \!+\! \nel, ~\label{eq:yel}
\end{align}
%------------------------------
respectively, where $\nik$ and $\nel$ are additive $\mathcal{CN} (0, 1)$ noise terms at the $k$-th IU and $\ell$-th EU, respectively.
%-----------------------------------------------------------------------
%
\subsection{Downlink SE and Average Harvested Energy}
%-----------------------------------------------------------------------
By invoking~\eqref{eq:yik} and using the bounding technique in~\cite{marzetta2016fundamentals,Hien:cellfree}, known as hardening bound, we derive a lower bound on the DL SE of  IU $k$. To this end, we first rewrite~\eqref{eq:yik} as
%-------------------
\begin{align}~\label{eq:yi:hardening}
    \yik &=  \mathrm{DS}_k  \xik +
    \mathrm{BU}_k \xik
     +\sum\nolimits_{k'\in\Kd\setminus k}
     \mathrm{IUI}_{kk'}
     \xikp
     \nonumber\\
   &\hspace{0em}
    + \sum\nolimits_{\ell\in\LL}
     \mathrm{EUI}_{k\ell}\xel +\nik,~\forall k\in\Kd,
\end{align}
%-------------------
where 
%----------------------
\begin{align}
 \mathrm{DS}_k  &=  \sum\nolimits_{m\in\MM}\sqrt{\rho a_m\etamkI} \Ex\Big\{\big(\gmkue\big)^H\wimk \Big\},~\label{eq:DS}
 \\
 \mathrm{BU}_k  &=  
 \sum\nolimits_{m\in\MM}
 \Big(\sqrt{\rho a_m\etamkI}\big(\gmkue\big)^H\wimk \nonumber\\
   &\hspace{2em}
   -\sqrt{\rho a_m\etamkI} \Ex\Big\{  \big(\gmkue\big)^H\wimk \Big\}  \Big),~\label{eq:BU}
 \\
 \mathrm{IUI}_{kk'} &= \sum\nolimits_{m\in\MM}
  \sqrt{\rho a_m\etamkpI}
   \big(\gmkue\big)^H\wimkp, ~\label{eq:IUI}
 \\
\mathrm{EUI}_{k\ell}  &=
\sum\nolimits_{m\in\MM}
   \sqrt{\rho a_m\etamlE}
   \big(\gmkue\big)^H\weml,~\label{eq:EUI}
\end{align}
%-----------------------
represent the strength of the desired signal ($\mathrm{DS}_k$),  the beamforming gain uncertainty ($\mathrm{BU}_k$), the interference caused by the $k'$-th IU, and the interference caused by the $\ell$-th EU, respectively.

We treat the sum of the second, third, fourth, and fifth terms in~\eqref{eq:yi:hardening} as ``effective noise". Since $\xik$ is independent of $\mathrm{BU}_k$, the first and the second terms are uncorrelated. Moreover, $\xik$ and  $\xikp$ are uncorrelated for any $k \neq k'$, thus, the first term in~\eqref{eq:yi:hardening} is uncorrelated with the third term.  Similarly, the first and fourth term  are uncorrelated as $\xik$ and  $\xel$ are uncorrelated. Therefore, the effective noise and the desired signal are uncorrelated. Accordingly, by invoking~\cite[Sec. 2.3.2]{marzetta2016fundamentals}, an achievable DL SE for IU $k$ can be written as
%-------------------
\begin{align}~\label{eq:SEk:Ex}
    \mathrm{SE}_k(\aaa, \ETAI, \ETAE)
      &\!=\!\!
      \Big(\!1\!- \!\frac{\tau}{\tau_c}\Big)
      \log_2
      \left(\!
       1\! +\! \SINRk (\aaa, \ETAI, \ETAE)
     \!\right),
\end{align}
%-------------------
where $ \ETAI = [\eta_{m1}^{\mathtt{I}}, \ldots, \eta_{mK_d}^{\mathtt{I}}]_{\vert m=1,\ldots, M}$, $ \ETAE = [\eta_{m1}^{\mathtt{E}}, \ldots, \eta_{mL}^{\mathtt{E}}]_{\vert m=1,\ldots, M}$,  while the effective SINR is given by %$\SINRk(\aaa,  \ETAI, \ETAE) =$
%--------------------
\begin{align}~\label{eq:SINE:general}
    &\SINRk(\aaa,  \ETAI, \ETAE) =\nonumber\\
   &\hspace{0em}
    \!\frac{
                 \big\vert  \mathrm{DS}_k  \big\vert^2
                 }
                 {  
                 \Ex\Big\{\! \big\vert  \mathrm{BU}_k  \big\vert^2\!\Big\} 
                 \!+\!
                 \sum_{k'\neq k}\!
                  \Ex\Big\{\! \big\vert \mathrm{IUI}_{kk'} \big\vert^2\!\Big\}
                  \!\! + \!
                  \sum_{\ell\in\mathcal{L}}\!
                 \Ex \Big\{\! \big\vert  \mathrm{EUI}_{k\ell} \big\vert^2\!\Big\}
                  \!\! +  1}.
\end{align}
%---------------------

To characterize the HE precisely, a NL EH model with the sigmoidal function is used. Therefore, the total HE at EU $\ell$ is given by
 %-------------------------
% \setcounter{equation}{22}
  \begin{align}~\label{eq:NLEH}
  \Phi^{\NL}_{\ell}\big(\aaa, \ETAI,  \ETAE\big) = \frac{\Psi\big(\mathrm{E}_{\ell}(\aaa,   \ETAI, \ETAE)\big) - \phi \Omega }{1-\Omega}, ~\forall \ell\in\LL,
 \end{align}
 %---------------------------
where  $\phi$ represents the maximum output DC power when the EH circuit reaches saturation, $\Omega=\frac{1}{1 + \exp(\xi \chi)}$ is a constant to guarantee a zero input/output response, while $\Psi\big(\mathrm{E}_{\ell}(\aaa,   \ETAI, \ETAE)\big)$  is the traditional logistic function, given by
 %------------------------
 \vspace{0.5em}
  \begin{align}~\label{eq:PsiEl}
     \Psi\big(\mathrm{E}_{\ell}(\aaa, \ETAI, \ETAE)\big) &\!=\!\frac{\phi}{1 + \exp\big(\!-\xi\big(\mathrm{E}_{\ell}(\aaa, \ETAI, \ETAE)\!-\! \chi\big)\!\big)},
 \end{align}
 %------------------------
where  $\xi$ and $ \chi$ are constant related parameters that  depend on the specific resistance, the capacitance and the circuit sensitivity. Moreover, $\mathrm{E}_{\ell}(\aaa, \ETAI, \ETAE)$ denotes the received energy for EH at $\ell$th EU, $\forall \ell\in\mathcal{L}$, which is given by 
%-----------------
\vspace{0.5em}
  \begin{align}~\label{eq:El_input}
    E_{\ell}\big(\!\aaa,\ETAI,   \ETAE\!\big)
      &\!=\!
     (\tau_c-\tau)\Snn
     \nonumber\\
   &\hspace{-4em}
     \times
     \Big( {\rho}\sum\nolimits_{m\in\MM}\!\sum\nolimits_{\ell'\in\LL}\!\!
   {(1-a_m)\etamlpE} \big\vert\big(\gmls\big)^{\!H}\wemlp\big\vert^2 
   \nonumber\\
   &\hspace{-4em}
   +\!{\rho}\!
   \sum\nolimits_{m\in\MM}\!\sum\nolimits_{k\in\Kd}
   \!\!{a_m\etamkI}\big\vert\big(\gmls\big)^{\!H}\wimk\big\vert^2\! +\! 1\Big).
\end{align}  

%-----------------

We denote the average of the HE at EU $\ell$ as 
%-------------------
\begin{align}~\label{eq:El_average}
Q_{\ell}(\aaa, \ETAI, \ETAE) \!=\Ex\big\{\mathrm{E}_{\ell}(\aaa,  \ETAI, \ETAE)\big\}.
\end{align}
%------------------

Before proceeding, by inspecting~\eqref{eq:NLEH}, we notice that $\Omega$ does not have any effect on the optimization problem. Therefore, we directly  consider $\Psi(\mathrm{E}_{\ell}(\aaa, \ETAI, \ETAE))$ to describe the HE at EU $\ell$. The inverse function of~\eqref{eq:PsiEl}, for $\phi\geq\Psi$, is given by
%---------------------------
\vspace{-0.1em}
\begin{align}~\label{eq:EinvPsi}
    \HEQoS (\Psi) =  \chi - \frac{1}{\zeta}\ln\bigg( \frac{\phi-\Psi}{\Psi}\bigg), \forall \ell. 
\end{align}
%---------------------------

The achievable DL SE in~\eqref{eq:SEk:Ex} and average HE in~\eqref{eq:El_average} are general and valid regardless of the precoding scheme used at the APs. We derive closed-form expressions for the proposed precoding scheme in the following subsection.

%-------------------------------------------------------
\vspace{-1em}
\subsection{Protective Partial Zero-Forcing Precoding Design}
%-------------------------------------------------------
From the general expressions for the SE and HE in~\eqref{eq:SEk:Ex} and~\eqref{eq:NLEH}, it is clear that, without an efficient precoding design at the APs, the simultaneous operation of I-APs and E-APs over the same frequency band would introduce co-channel interference, thereby degrading the performance of both IUs and EUs. To mitigate this, we now propose a precoding design for the I-APs and E-APs that fully cancels interference at the IUs (if channel estimation is perfect) while enabling the EUs to exploit the co-channel interference for enhanced EH. Notably, the literature has consistently highlighted that, in multiuser scenarios, interference signals can act as a valuable source for EH~\cite{Zhao:Access:2017}.
Building on this, we propose to utilize the PPZF precoding at the APs. The principle behind this design is that each I-AP employs PZF precoder to suppress the interference it causes to the IUs, while its interference with the EUs is exploited to power EUs to support their operation. On the other hand, the E-APs transmit energy signals to the EUs by using the maximum-ratio transmission (MRT) to maximize the amount of HE at the EUs. Nevertheless, IUs experience non-coherent interference from the energy signals transmitted to the EUs. To reduce this interference, MRT can be forced to take place in the orthogonal complement of the IUs' channel space. This design is called PMRT. 

We define the matrix of the channel estimates for $m$-th AP as $\hat{\qG}_m = \big[\Ghmu, \Ghms\big] \in\C^{N\times(K_d+L)}$, where  $\Ghmu = \big[\ghmonue, \ldots, \ghmKdue\big]$ denotes the estimate of all channels between AP $m$ and all IUs,  while $\Ghms= \big[\ghmons, \ldots, \ghmLs\big]$ is the estimate of all channels between AP $m$ and all EUs. Now, the PZF and PMRT precoder at the $m$-th AP towards IU $k$ and EU $\ell$, can be expressed as
%-------------------
\vspace{0.5em}
\begin{subequations}
 \begin{align}
    \wimk^{\PZF} &=\alpha^{\PZF}_{mk}
    { \Ghmu \Big(\big(\Ghmu\big)^H \Ghmu\Big) ^{-1} \qe_k^I}
    ,~\label{eq:wipzf}\\
        \weml^{\PMRT} &= \alpha^{\PMRT}_{m\ell}  \qB_m\Ghms\qe_{\ell}^{E},~\label{eq:wemrt}
\end{align}   
\end{subequations}
%-------------------
where $\qe_{k}^{I}$ ($\qe_{\ell}^{E}$) is the $k$-th column of $\qI_{K_d}$ ($\ell$-th column of $\qI_{L}$);
$\alpha^{\PZF}_{mk}=\sqrt{(N-K_d)\gamuemk}$ and $\alpha^{\PMRT}_{m\ell}=1/\sqrt{\big(N\!-\!K_d\big)\gamsml}$ denote the precoding normalization factors; $\qB_m$ denotes the projection matrix onto the orthogonal complement of $\Ghmu$, which is given by
%---------------------------
\vspace{-0.2em}
\begin{align}
  \qB_m  = \qI_{\Ntx}  - \Ghmu \Big( \big(\Ghmu\big)^H \Ghmu\Big)^{-1}  \big(\Ghmu\big)^H,
\end{align}
%---------------------------
which implies $\big(\ghmkue\big)^H \qB_m =\bf{0}$.

For given $\wimk^{\PZF}$ and $\weml^{\PMRT}$, in the following theorems, we provide closed-form expressions for the SE and average harvested RF power with the PPZF scheme.

\begin{Theorem}~\label{Theorem:SE:PPZF}
The ergodic SE for the $k$-th IU, achieved by the PPZF scheme is given in closed-form by~\eqref{eq:SEk:Ex}, where the effective SINR is given by~\eqref{eq:SINE:PPZF} at the top of the next page.
\end{Theorem}

\begin{proof}
See Appendix~\ref{Theorem:SE:PPZF:proof}.
\end{proof}
%--------------------
\begin{figure*}
\begin{align}~\label{eq:SINE:PPZF}
    \SINRk(\aaa, \ETAI, \ETAE) =
    \!\frac{
                  \rho \big(N-K_d\big)\Big(\sum_{m\in\MM}\sqrt{ a_m\etamkI \gamuemk}  \Big)^2
                 }
                 { \rho \sum_{m\in\MM}
                 \sum_{k'\in\Kd}
  a_m\etamkpI 
  \Big(\betamkue\!-\!\gamuemk\Big)
                  \! + \!
                 \rho
                 \sum_{m\in\MM}
                 \sum_{\ell\in\mathcal{L}}
   { (1\!-\!a_m)\etamlE} \Big(\betamkue\!-\!\gamuemk\Big)
                   \!+\!  1}.
\end{align}
  	\hrulefill
	\vspace{-4mm}
  \end{figure*}
%---------------------

\begin{Theorem}~\label{Theorem:RF:PPZF}
The average HE at EU $\ell$, achieved by the PPZF scheme, is given by
%---------------------------------------------------
\vspace{-0.2em}
\begin{align}~\label{eq:El_average:PPZF}
    &Q_{\ell} (\aaa, \ETAI, \ETAE)
 =  (\tau_c-\tau)\Snn
     \Big(
     {\rho}\big(N-K_d+1\big)   
     \nonumber\\
     &\times\sum\nolimits_{m\in\MM}\!
     {(1-a_m)\etamlE} \gamsml
     +
      {\rho}\!\!\!\sum_{m\in\MM}\sum_{\ell'\in\LL\setminus \ell}\!
   {(1\!-\!a_m)\etamlpE} \betamls 
     \nonumber\\
     &\hspace{0em}
   +\!
   {\rho}
   \sum\nolimits_{m\in\MM}\sum\nolimits_{k\in\Kd}\!
   {a_m\etamkI}\betamls \!+\! 1\Big).
\end{align}
%----------------------------------------------------
\end{Theorem}

\begin{proof}
The proof is omitted due to the space limit.
\end{proof}

%======================================================
\vspace{-1.5em}
\subsection{Problem Formulation}~\label{sec:problem:formulation}
%======================================================
We formulate two AP mode selection and power
control problems: 1)  average sum-SE maximization for the IUs and 2) EE maximization.
%=================================================
\subsubsection{Sum-SE Maximization}\label{sec:SE}
%-------------------------------
We aim at optimizing the AP mode selection vectors ($\aaa$) and power control coefficients ($\ETAI, \ETAE$) to maximize the sum-SE, under the constraints on per-IU SE, per-EU HE, and transmit power at each AP. More precisely,  we formulate an optimization problem as
%--------------------------------------------------------------------
\begin{subequations}\label{P:SE}
	\begin{align}
		\text{\textbf{(P1)}:}~\underset{\qx}{\max}\,\, \hspace{2em}&
		\sum\nolimits_{k\in \Kd} \SEk (\aaa,\ETAI, \ETAE)  
		\\
		\mathrm{s.t.} \,\,
		\hspace{2em}& \Ex\left\{\Phi_{\ell}^{\NL}\big(\aaa,\ETAI, \ETAE\big) \right\}\geq \Gamma_{\ell},~\forall \ell\in\mathcal{L},\label{P:sumSEmaximiz:c1}\\
		& \SEk (\aaa,  \ETAI, \ETAE)  \geq \SEQoS,~\forall k \in \Kd,\label{P:sumSEmaximiz:c2}\\
			&\sum\nolimits_{k\in\Kd}
        {\etamkI}\leq a_m,
        ~\forall m\in\MM,\label{P:sumSEmaximiz:c3}\\
              &\sum\nolimits_{\ell\in\mathcal{L}}
        {\etamlE}\leq 1-a_m,~\forall m\in\MM,\label{P:sumSEmaximiz:c4}\\
			& a_m \in\{ 0,1\},\label{P:sumSEmaximiz:c5}
		\end{align}
\end{subequations}
%--------------------------------------------------------------------
where $\qx = \{\aaa, \ETAI, \ETAE\}$; $\SEQoS$ is the minimum SE required by the $k$-th IU to guarantee the QoS in the network; $\Gamma_{\ell}$ is the minimum required harvested power at EU $\ell$. Regarding the above formulation, the following remark is in order. 
\begin{remark}
   
The power constraint at AP $m$ expressed as \eqref{P:sumSEmaximiz:c3}, whilst \eqref{P:sumSEmaximiz:c4} is our contribution towards an efficient formulation and thus deserves further explanation. In fact,  the power constraint at AP $m$ is naturally written as $a_m\sum_{k\in\Kd}{\etamkI}+(1-a_m)\sum_{\ell\in\mathcal{L}} {\etamlE} \leq 1$, which is obvious from \eqref{eq:Xm}. However, if we adopt this constraint in \eqref{P:SE}, there are some potential numerical issues. For instance, if $a_m=0$, this constraint can allow $\etamkI$ to take a large value which may cause some numerical issues for optimization methods to be developed. Instead, we introduce equivalent power constraints given in 
\eqref{P:sumSEmaximiz:c3} and \eqref{P:sumSEmaximiz:c4}. In this way, the power coefficients for information and power transfer are forced to zero according to the AP's operating mode.
\end{remark}

%=================================================
\subsubsection{EE Maximization Problem}
\label{sec:EE}
%=================================================
Although the SE is acknowledged as a crucial design factor for wireless networks, recent attention has shifted towards the EE, driven by growing interest in environmentally friendly wireless networks. Consequently, there is an interesting opportunity to explore the EE of  CF-mMIMO systems examined in this paper.

To formulate the total EE optimization, we first evaluate the power consumption of the system, which incorporates the power consumption of APs and fronthaul links during DL information transmission. Denote by $\Pbhm$ the power consumed by the fronthaul link between the CPU and AP $m$. Moreover, denote by $P_{\mathtt{D},k}$ the power consumption to run circuit components for the DL transmission at DL UE $k$. The total power consumption over the considered CF-mMIMO system is modeled as~\cite{Emil:TWC:2015:EE,ngo18TGN,Xu:TVT:2021}
%--------------------
\begin{align}~\label{eq:Ptotal}
P_\mathtt{total}(\aaa,  \ETAI, \ETAE)& =  \sum\nolimits_{m\in\mathcal{M}} P_m(\boldsymbol a,\ETAI)   + \sum\nolimits_{k\in\Kd} P_{\mathtt{D},k}\nonumber\\
     &\hspace{2em} + \sum\nolimits_{m\in\mathcal{M}} \Pbhm(\aaa,  \ETAI, \ETAE), 
\end{align}
%---------------------
where $P_m(\boldsymbol a,\ETAI)$ denotes the power consumption at AP $m$ that includes the power consumption of the transceiver chains and the power consumed for the DL transmission. The power consumption $P_m(\boldsymbol a,\ETAI)$ can be modeled as~\cite{Emil:TWC:2015:EE,ngo18TGN}
%---------------------
\vspace{-0.1em}
\begin{align}~\label{eq:Pm}
P_m (\boldsymbol a,\ETAI) 
&\! =\! a_m\left[\frac{\rho_{d}\Sn}{\zeta_m}\left(\sum\nolimits_{k\in\Kd}  \etamkI\right)
\!+\!N P_{\mathtt{cdl},m}\right]\!, 
\end{align}
%------------------------
where $0<\zeta_m\leq 1 $ is the power amplifier efficiency at the $m$-th AP;  $P_{\mathtt{cdl},m}$ is the internal power required to run the circuit components (e.g., converters, mixers, and filters) related to each antenna of AP $m$ for the DL transmissions.

Since the fronthaul links are used to transfer data between the APs and the CPU, their total power consumption is proportional to the aggregate data rate transmitted over each fronthaul link. 
Let $B$ be the system bandwidth. The fronthaul rate between AP $m$ and the CPU is
%--------------------
\begin{align}
R_m(\aaa,  \ETAI, \ETAE) =  B  a_m \sum\nolimits_{k\in\Kd}\SEk (\aaa,  \ETAI, \ETAE).
\end{align}
%---------------------

Denote by $P_{\mathtt{fdl},m}$  the fixed power consumption for the DL  transmission of each fronthaul, which is traffic-independent and may depend on the distances between the APs and the CPU and the system topology. Then, the power consumption of the fronthaul signal load to each I-AP $m$ is proportional to the fronthaul rate as \cite{ngo18TGN,bashar21TCOM}
%--------------------
\begin{align}~\label{eq:Pbhm}
\Pbhm(\aaa,  \ETAI, \ETAE) = a_m P_{\mathtt{fdl},m} + R_m(\aaa,  \ETAI, \ETAE)
P_{\mathtt{bt},m},
\end{align}
%---------------------
where $P_{\mathtt{bt},m}$ is the traffic-dependent fronthaul power (in Watt per bit/s).
By substituting~\eqref{eq:Pm} and~\eqref{eq:Pbhm} into~\eqref{eq:Ptotal}, we have
%--------------------
\begin{align}~\label{eq:Ptotal:final}
&P_\mathtt{total} (\aaa,  \ETAI, \ETAE)
\!=\!\!\! 
\sum_{m\in\mathcal{M}} \!\!\!
a_m\!\!\left[\rodsnz\left(\sum\nolimits_{k\in\Kd} \!\!\! \etamkI\right)
\!+\!N P_{\mathtt{cdl},m}\right]
\nonumber\\
     &\hspace{0em}
+\!\sum\nolimits_{k\in\Kd}\!\!\! P_{\mathtt{D},k}
+ \!\sum\nolimits_{m\in\mathcal{M}}\!\! \Big(a_m P_{\mathtt{fdl},m}\! + \!R_m(\aaa,  \ETAI, \ETAE)
P_{\mathtt{bt},m}\Big),
\nonumber\\
&= 
\sum\nolimits_{k\in\Kd}\! P_{\mathtt{D},k}\!+\!\sum\nolimits_{m\in\mathcal{M}} \!\!
a_m\!\bigg[\rodsnz\Big(\sum\nolimits_{k\in\Kd}  \etamkI\Big)
\nonumber\\
&\hspace{1em}+B P_{\mathtt{bt},m} \sum\nolimits_{k\in\Kd}\SEk (\aaa,  \ETAI, \ETAE) + P_m^{\mathtt{fixed}}
\bigg],
\end{align}
%-------------------------
where $P_m^{\mathtt{fixed}}=N P_{\mathtt{cdl},m} + P_{\mathtt{fdl},m}$.

The total EE (bit/Joule) is defined as the sum
throughput (bit/s) divided by the total power consumption
(Watt) in the network
%---------------
\vspace{-0.3em}
\begin{align}~\label{eq:EE:def}
   \EE(\aaa,  \ETAI, \ETAE)  = \frac{ B  \sum_{k\in\Kd}\SEk (\aaa,  \ETAI, \ETAE)}{P_\mathtt{total} (\aaa,  \ETAI, \ETAE)}.  
\end{align}
%----------------

Now, we formulate the total EE optimization, under the constraints
on the per-IU SE, per-EU HE, and transmit power at each
AP. More precisely, the optimization problem is formulated
as follows:
%--------------------------------------------------------------------
\begin{subequations}\label{P:EEmaximiz}
	\begin{align}
		\text{\textbf{(P2):}}~\underset{\qx}{\max}\,\, \hspace{1em}&~
		\EE (\aaa,  \ETAI, \ETAE)  
		\\
		\mathrm{s.t.} \,\,
		\hspace{2em}&
  \Ex\left\{\Phi_{\ell}^{\NL}\big(\aaa,\ETAI, \ETAE\big) \right\}\geq \Gamma_{\ell},~\forall \ell\in\mathcal{L},\label{P:EEmaximiz:c1}\\ 		
		 &\SEk (\aaa,  \ETAI, \ETAE)  \geq \SEQoS,~\forall k\in \Kd,\label{P:EEmaximiz:c2}\\
			&\sum\nolimits_{k\in\Kd}
        {\etamkI}\leq a_m,~\forall m\in\MM,\label{P:EEmaximiz:c3}\\
       &\sum\nolimits_{\ell\in\mathcal{L}}
        {\etamlE}\leq 1-a_m,~\forall m\in\MM,\label{P:EEmaximiz:c4}\\
		% & a_m + b_m =1,\\
		& a_m \in\{ 0,1\}.\label{P:EEmaximiz:c5}
		\end{align}
\end{subequations}
%-----------------------------------------------------------------

The formulated problems \textbf{(P1)} and \textbf{(P2)} are mixed Boolean (or binary) non-convex optimization problems, making them generally NP-hard. Finding a globally optimal solution to these problems is challenging not only from the perspective of algorithmic development but also due to the extremely high computational complexity involved. For clarity, let us focus on  \textbf{(P1)}. To handle the binary variables $a_m$, a common approach is the branch-and-bound method~\cite{stanford_ee364b_notes}, which fixes $a_m$ to either $0$ or $1$ at each iteration and relaxes other $a_m$ to the range $0\leq a_{m}\leq1$. However, in the worst case, this requires searching through all $2^{|\mathcal{M}|}$ possibilities, leading to an exponential increase in the complexity with the problem size. Even when the binary vector 
 $\mathbf{a}$ is fixed, the resulting problem of optimizing the power control coefficients remains non-convex due to the non-convex nature of the objective function and constraint in~\eqref{P:sumSEmaximiz:c1}. This further amplifies the difficulty of solving  \textbf{(P1)} and \textbf{(P2)}. It is worth noting that even without constraints~\eqref{P:sumSEmaximiz:c1} and ~\eqref{P:sumSEmaximiz:c2}, 
\textbf{(P1)} resembles the classical power control problem, which has been proven to be strongly NP-hard~\cite{Luo:JTSP:2009}. The NP-hardness of 
\textbf{(P1)} in this simplified scenario implies that no deterministic polynomial-time solutions are currently known. Consequently, it is widely accepted that the complexity of strongly NP-hard problems grows exponentially with the problem size.

%---------------------------------------------------------------------------
 \section{Proposed Solutions for Sum-SE and EE Optimization Problems}~\label{sec:optimization}
 %-------------------------------------------------------------------------------------
Before proceeding to solve the optimization problems, we focus on the constraint related to the minimum amount of HE at the EUs. Note that the constraint $\Ex\left\{\Phi_{\ell}^{\NL}\big(\aaa,\ETAI, \ETAE\big) \right\}\geq \Gamma_{\ell}$ in \text{\textbf{(P1)}} and \text{\textbf{(P2)}} is equivalent to
%---------------
\begin{align}~\label{eq:EH:constraint}
    \Ex\left\{\Psi\left(\mathrm{E}_{\ell}\big(\aaa,  \ETAI, \ETAE\big)\right)\right\}\geq \tilgmal,
\end{align}
%---------------
where $\tilgmal \triangleq (1-\Omega)\Gamma_{\ell} + \phi \Omega$. Due to the complicated nature of the logistic function, it is not easy to deal with the constraint in~\eqref{eq:EH:constraint}. In what follows, we use the approximation $\Ex\left\{\Psi\left(\mathrm{E}_{\ell}\big(\aaa,  \ETAI, \ETAE\big)\right)\right\} \approx \Psi\left(\Ex\left\{\mathrm{E}_{\ell}\big(\aaa,  \ETAI, \ETAE\big)\right\}\right)$.  Therefore, the constraint in~\eqref{eq:EH:constraint} is equivalent to
%---------------
\begin{align}~\label{eq:EH:ctj}
    \Psi\left(Q_{\ell}\big(\aaa,  \ETAI, \ETAE\big)\right)\geq \tilgmal.
\end{align}
%--------------- 

\begin{remark}
    The logistic function $\Psi\left(\mathrm{E}_{\ell}\big(\aaa,  \ETAI, \ETAE\big)\right)$ defined in~\eqref{eq:PsiEl} is a convex function of $\mathrm{E}_{\ell}(\aaa, \ETAI, \ETAE)$ for $\mathrm{E}_{\ell}(\aaa, \ETAI, \ETAE)\leq \chi$. Since WPT typically delivers low power to the input of the NL-EH circuits, we can assume that the logistic function is a convex function of the input energy. Accordingly, by applying Jensen's inequality, we have
$\Ex\left\{\Psi\left(\mathrm{E}_{\ell}\big(\aaa,  \ETAI, \ETAE\big)\right)\right\}\leq \Psi\left(\Ex\left\{\mathrm{E}_{\ell}\big(\aaa,  \ETAI, \ETAE\big)\right\} \right)$, and thus the left-hand-side of~\eqref{eq:EH:ctj} is a lower bound for the left-hand side of~\eqref{eq:EH:constraint}. Therefore,~\eqref{eq:EH:ctj} is equivalent to~\eqref{eq:EH:constraint}. 
\end{remark}

%=================================================
\vspace{-1.5em}
\subsection{Sum-SE Maximization}
\label{sec:sumSE}
%=================================================
By applying~\eqref{eq:EinvPsi} and~\eqref{eq:EH:ctj}, we can consider a conservative formulation, but more tractable, of~\eqref{P:SE} given by 
%--------------------------------------------------------------------
\begin{subequations}\label{P:SE:final}
	\begin{align}
		\text{\textbf{(P1.1)}:}~\underset{\qx}{\max}\,\, \hspace{0em}&
		\sum\nolimits_{k\in \Kd} \!\SEk (\aaa, \ETAI, \ETAE) 
		\\
		\mathrm{s.t.} \,\,
		\hspace{2em}& Q_{\ell}\big(\aaa, \ETAI, \ETAE\big)  \!\geq\! \HEQoS(\tilgmal),\!~\forall \ell \in \mathcal{L},\\
  &~\eqref{P:sumSEmaximiz:c2}-\eqref{P:sumSEmaximiz:c5}.
		\end{align}
\end{subequations}
%--------------------------------------------------------------------
Note that $\phi\geq\tilgmal$ results in $\phi\geq\Gamma_{\ell}$, which means that the maximum HE at each EU $\ell$ cannot exceed $\phi$, since $\phi$ is the maximum output DC power when the EH circuit reaches saturation. Our proposed method for solving \textbf{(P1)} is based on continuous relaxation, i.e., replacing \eqref{P:sumSEmaximiz:c5} by the continuous counterpart $0\leq a_m \leq 1, \forall m$. To this end, we observe that for any real number $a_m$, we have $(a_m\in\{0,1\}\Leftrightarrow a_m-a_m^2=0)\Leftrightarrow (a_m\in[0,1]\,~\&\,~a_m-a_m^2\leq0)$~\cite{Mohammadi:JSAC:2023}. This constraint can  directly be incorporated into the objective function of the optimization problem as a penalty term with parameter $c_1$.  Even with this relaxation, \textbf{(P1)} is still nonconvex and thus difficult to solve. To deal with the nonconvexity, we shall adopt the SCA method, which aims to solve a series of convex approximate problems of \textbf{(P1)}. By introducing the slack variables $t_k$, it is easy to see that a continuous relaxation of \textbf{(P1.1)} is equivalent to
%-------------------------
\begin{subequations}\label{P:SE:final}
	\begin{align}
		\text{\textbf{(P1.2)}:}~\underset{\qx,\qt }{\max}&\hspace{0.2em}  
		\prod\nolimits_{k\in \Kd} (1+t_k) 
  -c_1 \!\!\sum\nolimits_{m\in\MM}\!\!a_m(1\!-\!a_m)\label{P:SE:final:obj}
		\\
		\mathrm{s.t.} 
		& \quad Q_{\ell}\big(\aaa, \ETAI, \ETAE\big) \geq \HEQoS(\tilgmal),~\forall \ell\in\mathcal{L}, \label{P:SE:final:ct1}\\
  & \quad\SINRk (\aaa, \ETAI, \ETAE) \geq t_k, \forall k \in \Kd\label{P:SE:final:ct2},\\
		& \quad t_k  \geq 2^{\barSEQoS}-1,~\forall k \in \Kd,\label{P:SE:final:ct3}\\
			&\quad \sum\nolimits_{k\in\Kd}
        {\etamkI}\leq a_m^2,~\forall m\in\MM, \label{P:SE:final:ct4}\\
       & \quad \sum\nolimits_{\ell\in\mathcal{L}}\!\!
        {\etamlE}\leq 1-a_m^2,~\forall m\in\MM,\label{P:SE:final:ct5}\\
		& \quad 0 \leq a_m \leq 1,~\forall m\in\MM \label{P:SE:final:ct6},
		\end{align}
\end{subequations}
%-----------------------------------
where $\barSEQoS=\frac{\SEQoS}{1- {\tau}/{\tau_c}}$ and $\qt=\{t_1,\ldots,t_{K_d}\}$. Note that, we have replaced $a_m$ by $a_m^2$ in \eqref{P:sumSEmaximiz:c3} and \eqref{P:sumSEmaximiz:c4}  to obtain \eqref{P:SE:final:ct4} and~\eqref{P:SE:final:ct5}.  It is easy to see that this reformulation does not affect the optimality when $a_m \in \{0,1\}$. The benefit of adopting \eqref{P:SE:final} is that  $a_m$ tends to get close to a binary value, as seen in our numerical experiments.

It is clear that the difficulty in solving \textbf{(P1.2)} lies in the nonconvexity of \eqref{P:SE:final:obj}, \eqref{P:SE:final:ct1},~\eqref{P:SE:final:ct2}, and~\eqref{P:SE:final:ct4}. As mentioned above, to deal with these constraints, we apply SCA. Let us consider \eqref{P:SE:final:ct1} first, which is equivalent to
%---------------------------------
\begin{multline}
1+\rho\big(N-K_{d}+1\big)\sum_{m\in\MM}\gamsml\etamlE+\rho\sum_{m\in\MM}\sum_{\ell'\in\LL\setminus\ell}\betamls\etamlpE\\
+\rho\!\!\sum_{m\in\MM}\!\!a_{m} u_{m\ell}
\geq\frac{\HEQoS(\tilgmal)}{(\tau_{c}-\tau)\Snn},
\end{multline}
%-------------------------------
where
\[
\uml\triangleq\betamls\sum_{k\in\Kd}\etamkI-\big(N-K_{d}+1\big)\gamsml\etamlE-\betamls\sum_{\ell'\in\LL\setminus\ell}\etamlpE.
\]
%-----------------------------------
Note that $\uml$ is not treated as a new optimization variable,
but as an ``expression holder''. Then, \eqref{P:SE:final:ct1} is equivalent to
%------------------------
\vspace{-0.2em}
\begin{align}
&4\rho\sum\nolimits_{m\in\MM}\sum\nolimits_{\ell'\in\LL\setminus\ell}\betamls\etamlpE+\!\rho\!\sum\nolimits_{m\in\MM}\!(a_{m}\!+\!\uml)^{2}
\nonumber\\
&+4\rho\big(N-K_{d}+1\big)\sum\nolimits_{m\in\MM}\!\!\gamsml\etamlE
+4\nonumber\\
&\hspace{2em}\geq\frac{4\HEQoS(\tilgmal)}{(\tau_{c}-\tau)\Snn}+\rho\sum\nolimits_{m\in\MM}(a_{m}-\uml)^{2}.
\end{align}

%------------------------
To facilitate the description, we use a superscript $(n)$ to denote
the value of the involving variable produced after $(n-1)$ iterations
($n\geq0$). In light of SCA, the above constraint can be approximated
by the following convex one
%-----------------------------
\vspace{-0.2em}
\begin{multline}
4\rho\big(N\!-\!K_{d}+1\big)\!\sum\nolimits_{m\in\MM}\etamlE\gamsml+4\rho\sum_{m\in\MM}\sum_{\ell'\in\LL\setminus\ell}\etamlpE\betamls\\
+\rho\sum\nolimits_{m\in\MM}(a_{m}^{(n)}+\umln)\Bigl(2(a_{m}+\uml)-a_{m}^{(n)}-\umln\Bigr)+4\\
\geq\frac{4\HEQoS(\tilgmal)}{(\tau_{c}-\tau)\Snn}+\rho\sum\nolimits_{m\in\MM}(a_{m}-\uml)^{2}, \label{eq:energy:approx}
\end{multline}
%-------------------------------
where $\umln=\betamls\sum_{k\in\Kd}{\etamkIn}-\big(N-K_{d}+1\big)\gamsml{\etamlEn}-\betamls\sum_{\ell'\in\LL\setminus\ell}{\etamlEn}$
and we have used the following inequality
%-------------------
\vspace{-1em}
\begin{equation}~\label{eq:x2}
x^{2}\geq 
x_{0}(2x-x_{0}),
\end{equation}
%--------------------------
and replaced $x$ and $x_{0}$ by $a_{m}+\uml$ and $a_{m}^{(n)}+\uml^{(n)}$,
respectively.

Now, we turn our attention to~\eqref{P:SE:final:ct2}.
Let $\nu_{mk}=\betamkue\!-\!\gamuemk$, $\etamI\triangleq\sum_{k'\in\K}\etamkpI$,
and $\etamE\triangleq\sum_{\ell\in\mathcal{L}}\etamlE$. Then,~\eqref{P:SE:final:ct2} can be rewritten as
%--------------------------
\vspace{-0.4em}
\begin{multline}
\frac{\rho\big(N-K_{d}\big)\Big(\sum_{m\in\MM}\sqrt{a_{m}\etamkI\gamuemk}\Big)^{2}}{t_{k}}\geq\\
\sum\nolimits_{m\in\MM}\rho\nu_{mk}a_{m}\left(\etamI\!-\!\etamE\right)\!+\!\sum\nolimits_{m\in\MM}\rho\nu_{mk}\etamE\!+\!1.
\end{multline}
%----------------

It is easy to see that we can further rewrite \eqref{P:SE:final:ct2} as 
%------------------
\vspace{-0.2em}
\begin{multline}
\frac{\rho\big(N-K_{d}\big)\Big(\sum_{m\in\MM}\sqrt{a_{m}\etamkI\gamuemk}\Big)^{2}}{t_{k}}\\
+\frac{1}{4}\sum\nolimits_{m\in\MM}\rho\nu_{mk}\bigl(a_{m}-\etamI+\etamE\bigr)^{2}\geq\\
% \frac{1}{4}\!\sum\nolimits_{m\in\MM}\!\rho\nu_{mk}\bigl(a_{m}\!+\!\etamI\!-\!\etamE\bigr)^{\!2}\!+\!\!\sum\nolimits_{m\in\MM}\!\rho\nu_{mk}\etamE\!+\!1.
\frac{1}{4}\!\sum\nolimits_{m\in\MM}\!\rho\nu_{mk}\big[\bigl(a_{m}\!+\!\etamI\!-\!\etamE\bigr)^{\!2}+\etamE \big]+\!1.
\end{multline}
%------------------

It is clear that we need to find a concave lower bound of the left
hand side of the above inequality. To this end, we note that the function
$x^{2}/y$ is convex for $y>0$, and thus, the following inequality
holds
%----------------------------
\vspace{-0.2em}
\begin{align}~\label{eq:x2/y}
\frac{x^{2}}{y}\geq
% \frac{x_{0}^{2}}{y_{0}}+2\frac{x_{0}}{y_{0}}(x-x_{0})-\frac{x_{0}^{2}}{y_{0}^{2}}(y-y_{0})=
\frac{x_{0}}{y_{0}}\Big(2x-\frac{x_{0}}{y_{0}}y\Big),
\end{align}
%----------------------------
which is obtained by linearizing $x^{2}/y$ around $x_{0}$ and $y_{0}$.
From the above inequality, it follows that
%----------------------------
\begin{equation}
\frac{\Big(\sum_{m\in\MM}\!\!\sqrt{a_{m}\etamkI\gamuemk}\Big)^{\!\!2}}{t_{k}}\!\geq \!q_{k}^{(n)}\Bigl(2\!\sum_{m\in\MM}\!\!\!\!\sqrt{\gamuemk a_{m}\etamkI}\!-\!q_{k}^{(n)}t_{k}\Bigr),
\end{equation}
%----------------------------
where $q_{k}^{(n)}=\Bigl(\sum_{m\in\MM}\sqrt{\gamuemk a_{m}^{(n)}\etamkIn}\Bigr)/t_{k}^{(n)}$.
Thus, by applying~\eqref{eq:x2} and~\eqref{eq:x2/y}, we can approximate \eqref{P:SE:final:ct2} by 
%----------------------------
\begin{align}
&\rho\big(N-K_{d}\big)q_{k}^{(n)}\biggl(2\sum\nolimits_{m\in\MM}\sqrt{\gamuemk a_{m}\etamkI}-q_{k}^{(n)}t_{k}\biggr)\nonumber\\
&
+\frac{1}{4}\sum\nolimits_{m\in\MM}\rho\nu_{mk}z_{m}^{(n)}\bigl(2z_{m}-z_{m}^{(n)}\bigr) \geq
%\sum\nolimits_{m\in\MM}\rho\nu_{mk}\etamE
\nonumber\\
&
+\frac{1}{4}\sum\nolimits_{m\in\MM}\big[\rho\nu_{mk}\bigl(a_{m}+\etamI-\etamE\bigr)^{2}+\etamE\big]+1,\label{eq:SINR:approx}
\end{align}
%-------------------
where $z_{m}=a_{m}+\etamE-\etamI$ and
$z_{m}^{(n)}=a_{m}^{(n)}+\etamEn-\etamIn$. We remark that the above constraint is convex which is due to the concavity of $\sqrt{ a_{m}\etamkI}$. From the above discussions and by invoking the upper bound~\eqref{eq:x2} to deal with the non-convex term in~\eqref{P:SE:final:obj} and~\eqref{P:SE:final:ct4},  we arrive at the following approximate convex problem
%------------------------------------------
\begin{subequations}\label{P:SE:relax:approx}
	\begin{align}
		\text{\textbf{(P1.3):}}~\underset{\qx,\qt}{\max}&\quad  
		\prod\nolimits_{k\in \Kd} (1+t_k) 
  \nonumber\\
  %-c_1\sum\nolimits_{l\in \mathcal{L}} \lambda_l\nonumber\\
  %&
  &-c_1\sum\nolimits_{m\in\MM}\Big(a_m-a_m^{(n)}\big(2a_m\!-\!a_m^{(n)}\big)\Big)
		\\
		\mathrm{s.t.}  
		& \quad \eqref{eq:energy:approx},~\forall \ell \in \mathcal{L},\\
  &\quad \eqref{eq:SINR:approx},~\forall k \in \Kd, 
            \\
		% & \quad t_k  \geq 2^{\barSEQoS}-1,~\forall k \in \Kd\\
			&\quad \sum\nolimits_{k\in\Kd}
        {\etamkI}\leq a_m^{(n)}\big(2a_m\!-\!a_m^{(n)}\big),
        % \nonumber\\
        % &\hspace{10em}
        ~\forall m\in\MM,~\label{P:SE:relax:approx:Ct3} \\
       %& \quad \sum\nolimits_{\ell\in\mathcal{L}}\!\!
        %{\etamlE}\leq 1-a_m^2,~\forall m\in\MM,\\
        &\quad~\eqref{P:SE:final:ct3},~\eqref{P:SE:final:ct5},
        ~\eqref{P:SE:final:ct6}.
		%& \quad 0 \leq a_m \leq 1,~\forall m\in\MM.
		\end{align}
\end{subequations}
%------------------------------------------

The convex optimization problem~\eqref{P:SE:relax:approx} can be efficiently solved by using existing standard convex optimization packages, such as CVX~\cite{cvx}. We iteratively solve problem~\eqref{P:SE:relax:approx} until the relative reduction of the objective function's value in~\eqref{P:SE:relax:approx} falls below the predefined threshold, $\epsilon$. 
%=================================================
\vspace{-0.5em}
\subsection{Energy Efficiency Maximization Problem}
\label{sec:EE}
%=================================================
The optimization problem~\eqref{P:EEmaximiz} is non-convex, due to the non-convexity of the objective function and constraints~\eqref{P:EEmaximiz:c1} and~\eqref{P:EEmaximiz:c2}. We start with the objective function in~\eqref{eq:EE:def}, which is explicitly written as
%--------------------
\vspace{-0.5em}
\begin{align}~\label{eq:EEinv}
&\frac{1}{\EE(\aaa,  \ETAI, \ETAE)}=\sum\nolimits_{m\in\mathcal{M}} 
        a_m P_{\mathtt{bt},m}+\frac{\sum_{k\in\Kd} P_{\mathtt{D},k}}{B\sum_{k\in\Kd} v_k}\nonumber\\
&\hspace{3em}+\frac{ \sum_{m\in\mathcal{M}} 
a_m\left[\rodsnz\left(\sum_{k\in\Kd}  \etamkI\right)
 + P_m^{\mathtt{fixed}}
\right]
}{B\sum_{k\in\Kd} v_k},
\end{align}
%----------------------
where $v_k$ is an auxiliary variable. Note that the first two terms in~\eqref{eq:EEinv} are convex, and thus the nonconvexity of the objective is due to the last term which is equivalently expressed as 
%---------------
\begin{align}
&\frac{ \sum_{m\in\mathcal{M}} 
\big(a_m+\rodsnz\left(\sum_{k\in\Kd}  \etamkI\right)
 + P_m^{\mathtt{fixed}}\big)^2
}{4B\sum_{k\in\Kd} v_k} \nonumber\\
&- \frac{ \sum_{m\in\mathcal{M}} 
\big(a_m-\rodsnz\left(\sum_{k\in\Kd}  \etamkI\right)
 - P_m^{\mathtt{fixed}}\big)^2
}{4B\sum_{k\in\Kd} v_k}.    
\end{align}
%-------------------------------
The second term is concave. Then, we apply SCA to approximate the second term in the above expression. By using~\eqref{eq:x2/y}, we have
%------------------
\begin{align}
  &\frac{  
\left(a_m-\rodsnz\left(\sum_{k\in\Kd}  \etamkI\right)
 - P_m^{\mathtt{fixed}}\right)^2
}{4B\sum\nolimits_{k\in\Kd} v_k} 
\geq \\
&\hspace{0em}
2\omega_m^{(n)} \bigg(a_m\!-\rodsnz\!\sum\nolimits_{k\in\Kd}\!\! \etamkI
 \!-\! P_m^{\mathtt{fixed}} \!-\! 2B\omega_m^{(n)} \sum\nolimits_{k\in\Kd}\!\! v_k \!\bigg),\nonumber
\end{align}
%---------------------------
where $\omega_m^{(n)}= \frac{a_m^{(n)}-\rodsnz\sum_{k\in\Kd}  \etamkIn
 - P_m^{\mathtt{fixed}}}{4B\sum_{k\in\Kd} v_k^{(n)}}$. Now, by using continuous relaxation, the EE maximization problem can be written as
 %---------------------------
 \vspace{-0.2em}
 \begin{subequations}\label{P:EEmaximiz:P42}
	\begin{align}
		\text{\textbf{(P2.1):}}~\underset{\qx,\qv}{\min}\,\, \hspace{0.05em}&~
         \tilde{f}_{\EE}(\aaa, \ETAI, \ETAE,\qv)
         % -\!c_1\sum\nolimits_{l\in \mathcal{L}}\lambda_{\ell} 
         \nonumber\\
         &\hspace{-1em}-c_1 \sum\nolimits_{m\in\MM}a_m-a_m^{(n)}\big(2a_m-a_m^{(n)}\big)
        		\\
		\mathrm{s.t.} \,\,
		\hspace{1em}&
   Q_{\ell}\big(\aaa, \ETAI, \ETAE\big)  \geq \HEQoS(\tilgmal),~\forall \ell\in \mathcal{L},\label{P:EEmaximiz:P42:c1}\\
   &
		 \SINRk (\aaa,  \ETAI, \ETAE)  \geq t_k,~\forall k\in \Kd,\label{P:EEmaximiz:P42:c2}\\
          & t_k \geq 2^{v_k}-1
          ,~\forall k\in \Kd,\label{P:EEmaximiz:P42:c3}\\
			&v_k  \geq \SEQoS,~\forall k\in \Kd,\label{P:EEmaximiz:P42:c4}\\
   &~\eqref{P:SE:final:ct5},~\eqref{P:SE:final:ct6},~\eqref{P:SE:relax:approx:Ct3},
		\end{align}
\end{subequations}
%---------------------
where 
%----------------
\vspace{-0.2em}
\begin{align}
    &\tilde{f}_{\EE}(\aaa, \ETAI, \ETAE,\qv) = \nonumber\\
    &\frac{ \sum_{m\in\mathcal{M}} 
\Big(a_m+\rodsnz\left(\sum_{k\in\Kd}  \etamkI\right)
 + P_m^{\mathtt{fixed}}\Big)^2
}{4B\sum_{k\in\Kd} v_k} 
  +\frac{\sum_{k\in\Kd} P_{\mathtt{D},k}}{B\sum_{k\in\Kd} v_k} \nonumber\\
    &
    +\!\sum\nolimits_{m\in\mathcal{M}}\!\! 
        a_m P_{\mathtt{bt},m}
\!- \! \sum\nolimits_{m\in\mathcal{M}}\!\! 
 \Big( 2\omega_m^{(n)} \Big(a_m\!-\rodsnz\!\sum\nolimits_{k\in\Kd}\!\!\etamkI
 \nonumber\\
&\hspace{6em}
- P_m^{\mathtt{fixed}} - 2B\omega_m^{(n)} \sum\nolimits_{k\in\Kd} v_k \Big)  \Big).
\end{align}
%-----------------
%quad_over_lin() 
By applying SCA to deal with the non-convex constraints~\eqref{P:EEmaximiz:P42:c1} and~\eqref{P:EEmaximiz:P42:c2}, we get the following convex optimization problem 
%---------------------------
 \begin{subequations}\label{P:EEmaximiz:P43}
	\begin{align}
		\text{\textbf{(P2.2):}}~\underset{\qx,\qv}{\min}\,\, \hspace{0.1em}&~
         \tilde{f}_{\EE}(\aaa, \ETAI, \ETAE,\qv)  -\nonumber\\
         &\hspace{-1em}-\!c_1 \!\!\sum\nolimits_{m\in\MM}\!\!a_m\!-\!a_m^{(n)}\big(2a_m\!-\!a_m^{(n)}\big)
        		\\
		\mathrm{s.t.} \,\,
		\hspace{0.1em}
          &  \eqref{eq:energy:approx},~\forall \ell \in \LL,\\ 
		 &  \eqref{eq:SINR:approx},~\forall k\in \Kd,\label{P:EEmaximiz:P43:c1}\\         &\eqref{P:SE:final:ct5},~\eqref{P:SE:final:ct6},~\eqref{P:SE:relax:approx:Ct3}.
		\end{align}
\end{subequations}
%---------------------

%%%%%%%%%%%%%%%%%%%%%%%%%%%%%%%
\vspace{-0.8em}
\section{Extensions and Benchmarks}~\label{sec:extension}
%%%%%%%%%%%%%%%%%%%%%%%%%%%%%%%
We can expand our optimization framework to include other objective functions. Specifically, we consider objective functions that prioritize the functionality of the EUs. Therefore, in this section, we first address the sum-HE maximization and max-min EU fairness problems. Additionally, to assess the effectiveness of our proposed SWIPT scheme for CF-mMIMO systems, we introduce benchmark schemes for comparison in the numerical results of Section~\ref{sec:num}.
%%%%%%%%%%%%%%%%%%%%%%%%%%%%%%%
\vspace{-0.8em}
\subsection{Extensions}
%%%%%%%%%%%%%%%%%%%%%%%%%%%%%%%
%=================================================
\subsubsection{Sum-HE Maximization}
\label{sec:SE}
%=================================================
For this design scheme, we aim at optimizing the AP mode assignment vectors ($\aaa$) and power control coefficients ($\ETAI, \ETAE$) to maximize the total HE, under the constraints on  minimum power requirements at the EUs, per-IU SE constraints, and transmit power at each APs.  The optimization problem is mathematically formulated as
%--------------------------------------------------------------------
\begin{subequations}\label{P:SHE:max}
	\begin{align}
		\text{\textbf{(P3):}}~\underset{\qx}{\max}\,\, \hspace{1em}&
		\sum\nolimits_{\ell\in \mathcal{L}} 
        \Ex\left\{\Phi_{\ell}^{\NL}\big(\aaa,\ETAI, \ETAE\big)\right\}
		\\
		\mathrm{s.t.} \,\,
		\hspace{2em}& \Ex\left\{\Phi_{\ell}^{\NL}\big(\!\aaa,\ETAI, \ETAE\big)\right\} \!\geq \!\Gamma_{\ell},\!\!~\forall \ell\in\mathcal{L},~\label{eq:P:SHE:max:ct1}\\
		&\SEk (\aaa,  \ETAI, \ETAE) \! \geq\! \SEQoS\!,~\forall k \in \Kd,~\label{eq:P:SHE:max:ct2}\\
		&\sum\nolimits_{k\in\Kd}
        {\etamkI}\leq a_m,
        ~\forall m\in\MM,\label{eq:P:SHE:max:ct3}\\
        &\sum\nolimits_{\ell\in\mathcal{L}}
        {\etamlE}\leq 1-a_m,~\forall m\in\MM,\label{eq:P:SHE:max:ct4}\\
		& a_m \in\{ 0,1\}.\label{eq:P:SHE:max:ct5}
		\end{align}
\end{subequations}
%--------------------------------------------------------------------

By using the auxiliary variables $e_{\ell}$, $\ell\in\LL$, and relaxing the binary variables, $a_m$, the optimization problem \text{\textbf{(P3)}} can be recast as
%--------------------------------------------------------------------
\vspace{-0.3em}
\begin{subequations}\label{P:SHE:max2}
	\begin{align}
		\text{\textbf{(P3.1):}}~\underset{\qx,\qe}{\max}\,\, \hspace{0.2em}&
		\sum\nolimits_{\ell\in \mathcal{L}} 
        e_{\ell}\nonumber\\
        &-\!c_1 \!\sum\nolimits_{m\in\MM}\!a_m\!-\!a_m^{(n)}\big(2a_m\!-\!a_m^{\!(n)}\big)
		\\
		\mathrm{s.t.} \,\,
		\hspace{2em}& 
         \Ex\left\{\Phi_{\ell}^{\NL}\big(\aaa,\ETAI, \ETAE\big)\right\}
         \!\geq\! e_{\ell},~\forall \ell\!\in\!\mathcal{L}\!,~\label{eq:P:SHE:max2:ct1}\\
        &e_{\ell} \geq \Gamma_{\ell},~\forall \ell\in\mathcal{L},~\label{eq:P:SHE:max2:ct2}\\
		&\SEk (\aaa,  \ETAI, \ETAE) \! \geq\! \SEQoS,~\forall k \!\in\! \Kd,~\label{eq:P:SHE:max2:ct3}\\
  &~\eqref{P:SE:final:ct5},~\eqref{P:SE:final:ct6},~\eqref{P:SE:relax:approx:Ct3},
		\end{align}
\end{subequations}
%--------------------------------------------------------------------
where $\qe=\{e_1,\ldots,e_{L}\}$.  Now, we deal with the two non-convex constraints in~\eqref{eq:P:SHE:max2:ct1} and~\eqref{eq:P:SHE:max2:ct2} as follows. 
In light of SCA, constraint~\eqref{eq:P:SHE:max2:ct1} can be approximated
by the following convex one
%-----------------------------
\begin{multline}
4\rho\big(N-K_{d}+1\big)\sum_{m\in\MM}\etamlE\gamsml+4\rho\sum_{m\in\MM}\sum_{\ell'\in\LL\setminus\ell}\etamlpE\betamls\\
+\rho\sum_{m\in\MM}(a_{m}^{(n)}+\umln)\Bigl(2(a_{m}+\uml)-a_{m}^{(n)}-\umln\Bigr)+4\\
\geq\frac{4 \tilde{\Xi}_{\ell}\left((1-\Omega)e_{\ell} + \phi \Omega\right)}{\eta(\tau_{c}\!-\!\tau)\Snn} +\rho\sum_{m\in\MM}(a_{m}-\uml)^{2}, \label{eq:energy:approx:SEH}
\end{multline}
%---------------------------------------
where $\tilde{\Xi}_{\ell}(x)$ is a convex upper bound for ${\HEQoS}(x) = \chi - \frac{1}{\zeta}\ln\Big( \frac{\phi-x}{x}\Big) $ which is given by 
%-------------------------
\begin{align}~\label{eq:conuppboun}
   \tilde{\Xi}_{\ell} (x)\!\triangleq\! \chi \!-\!\frac{1}{\zeta}\bigg(
   \ln\big( \frac{\phi-x}{x^{(n)}}\big) 
   % \nonumber\\
   % &
   %-\!\ln\big(x^{(n)}\big) 
   \!-\! \frac{1}{x^{(n)}}\big(x-x^{(n)}\big) \bigg).
\end{align}
%---------------------------------

Following a similar approach as in Section~\ref{sec:sumSE} to deal with the non-convex constraint~\eqref{P:SE:final:ct2}, the non-convex constraint~\eqref{eq:P:SHE:max2:ct3} can be approximated by the following constraint
%----------------------------
\vspace{-0.5em}
\begin{multline}
\frac{\rho\big(N-K_{d}\big)}{2^{\barSEQoS}-1}p_{k}^{(n)}\biggl(2\sum\nolimits_{m\in\MM}\sqrt{\gamuemk a_{m}\etamkI}-p_{k}^{(n)}\biggr)\\
+\frac{1}{4}\sum\nolimits_{m\in\MM}\rho\nu_{mk}z_{m}^{(n)}\bigl(2z_{m}-z_{m}^{(n)}\bigr)\\
\geq\frac{1}{4}\sum\nolimits_{m\in\MM}\rho\nu_{mk}\bigl(a_{m}\!-\!\etamI\!+\!\etamE\bigr)^{2}\!+\sum\nolimits_{m\in\MM}\rho\nu_{mk}\etamE\!+\!1,\label{eq:SINR:approx:SHE1}
\end{multline}
%----------------------------
where $p_{k}^{(n)}=\sum_{m=1}^{M}\sqrt{\gamuemk a_{m}^{(n)}\etamkIn}$. Therefore, problem \textbf{(P3)} is equivalent to
%-----------------------------------------
\begin{subequations}\label{P:SHE:max3}
	\begin{align}
		\text{\textbf{(P3.2):}}~\underset{\qx,\qe}{\max}\,\, \hspace{.2em}&
		\sum\nolimits_{\ell\in\LL}e_{\ell}\nonumber\\
        &-\!c_2 \!\sum\nolimits_{m\in\MM}\!a_m\!-\!a_m^{(n)}\big(2a_m\!-\!a_m^{\!(n)}\big)
		\\
		\mathrm{s.t.} \,\,
  		\hspace{2em} &~\eqref{eq:P:SHE:max2:ct2},~\eqref{P:SE:final:ct5},~\eqref{P:SE:final:ct6},~\eqref{P:SE:relax:approx:Ct3},~\eqref{eq:energy:approx:SEH},~\eqref{eq:SINR:approx:SHE1}.
		\end{align}
\end{subequations}
%------------------------

%=================================================
\subsubsection{Max-Min Average HE Optimization}
\label{sec:SE}
%=================================================
The AP mode selection and power control aim to maximize the worst-case average harvested RF power, $\text{min}_{\ell}~\Ex\left\{\Phi_{\ell}^{\NL}\big(\aaa,  \ETAI, \ETAE\big)\right\}$, with a constraint on the minimum SE and power requirements of the IUs and EUs, respectively.  Therefore, our constrained problem is expressed as
%--------------------------------------------------------------------
\vspace{-0.3em}
\begin{subequations}\label{P1:MaxMinEnergy}
	\begin{align}
		\text{\textbf{(P4):}}~\underset{\qx}{\max}\,\, \hspace{1em}&
		\underset{\ell}{\min}~
		\Ex\left\{\Phi_{\ell}^{\NL}\big(\aaa,  \ETAI, \ETAE\big)\right\}  
		\\
		\mathrm{s.t.} \,\,
		\hspace{2em}&
% 		Q_{\ell}\big(\aaa, \bb,\ETAI, \ETAE\big)  \geq \HEQoS,~\forall \ell\\
		 \SEk (\aaa,  \ETAI, \ETAE)  \geq \SEQoS,~\forall k\in \Kd,\label{P1:MaxMinEnergy:ct1}\\
			&\sum\nolimits_{k\in\Kd}
        {\etamkI}\leq a_m,~\forall m\in\MM,\label{P1:MaxMinEnergy:ct2}\\
       &\sum\nolimits_{\ell\in\mathcal{L}}
        {\etamlE}\leq 1-a_m,~\forall m\in\MM,\label{P1:MaxMinEnergy:ct3}\\
		% & a_m + b_m =1,\\
		& a_m \in\{ 0,1\}.\label{P1:MaxMinEnergy:ct4}
		\end{align}
\end{subequations}
%-----------------------------------------------------------------

Let us define the auxiliary the variable $t \triangleq \underset{\ell}{\min}~~\Ex\left\{\Phi_{\ell}^{\NL}\big(\aaa,  \ETAI, \ETAE\big)\right\}$. Then, the optimization problem (\textbf{P4}) is recast as 
%--------------------------------------------------------------------
\begin{subequations}\label{P:MaxMinEnergy:new}
	\begin{align}
		\text{\textbf{(P4.1):}}~\underset{\qx, t}{\max}\,\, \hspace{1em}& t  
		\\
		\mathrm{s.t.} \,\,
		\hspace{2em}&
       Q_{\ell}\big(\aaa,  \ETAI, \ETAE\big) \geq \HEQoS(t),~\forall \ell\in\mathcal{L},
       \\
       &~\eqref{P1:MaxMinEnergy:ct2},~\eqref{P1:MaxMinEnergy:ct3},~\eqref{P1:MaxMinEnergy:ct4}.
		\end{align}
\end{subequations}
%-----------------------------------------------------------------

The constraints of the optimization problem~\eqref{P:MaxMinEnergy:new}  are similar to~\eqref{P:SHE:max2}. Thus, applying similar steps, we get the following convex optimization problem
%-----------------------------------------
\begin{subequations}\label{P:MaxMinEnergy:new:solution}
	\begin{align}
		\text{\textbf{(P4.2):}}~\underset{\qx, t}{\max}\,\, \hspace{.2em}&
		t -\!c_2 \!\sum\nolimits_{m\in\MM}\!a_m\!-\!a_m^{(n)}\big(2a_m\!-\!a_m^{\!(n)}\big)
		\\
		\mathrm{s.t.} \,\,
              \hspace{2em}& ~\eqref{eq:energy:approx:SEH},~\forall \ell\in\mathcal{L},~\label{P:MaxMinEnergy:c2}
              \\
  		\hspace{2em} & \hspace{0em}~\eqref{eq:SINR:approx:SHE1}, \forall k \in \Kd,~\label{P:MaxMinEnergy:c1} 
               \\	&~\eqref{P:SE:final:ct5},~\eqref{P:SE:final:ct6},~\eqref{P:SE:relax:approx:Ct3},
		\end{align}
\end{subequations}
%------------------------
where, in~\eqref{P:MaxMinEnergy:c2}, $\tilde{\Xi}_{\ell}\left((1-\Omega)e_{\ell} + \phi \Omega\right)$ is replaced by $\tilde{\Xi}_{\ell}\left((1-\Omega)t + \phi \Omega\right)$.

\begin{remark}
     Extending the analysis in this paper to scenarios involving non-orthogonal pilots, such as high-mobility environments or ultra-dense networks with numerous users, requires specific modifications. In particular, an additional term must be incorporated into the denominator of the SINR expression in~\eqref{eq:SINE:PPZF} and another into the average input harvested energy in~\eqref{eq:El_average:PPZF}. Moreover, the statistical characterization of the channel estimates in Subsection II-A must also be updated. These changes will inevitably impact the SE of the IUs and the HE of the EUs. Nevertheless, it is important to emphasize that the proposed optimization methodologies 
     — (\textbf{P1.3}) for sum-SE maximization, (\textbf{P2.2}) for EE maximization, (\textbf{P3.2}) for sum-HE maximization, and (\textbf{P4.2}) for max-min average HE maximization — remain applicable to the non-orthogonal pilots' scenario.
\end{remark}

%=================================================
\vspace{-2em}
\subsection{Benchmarks}
\label{sec:benchmark}
%=================================================
%-----------------------------------
\subsubsection{Random AP Operation Mode Selection without Power Control (Benchmark 1)}
%------------------------------------------------------------------------------------------
We assume that the APs' operation mode selection parameters ($\aaa$) are randomly assigned and no power control is performed at the APs. PZF and PMRT precoding  are applied to the I-APs and E-APs, respectively. Moreover, in the absence of power control, both the E-APs and I-APs transmit at full power, i.e., at the $m$-th AP, power coefficients are the same and $\etamkI = \frac{1}{K_d}$, $\forall k\in\mathcal{K}_d$  and $\etamlE = \frac{1}{L}$, $\forall \ell\in\mathcal{L}$.

%--------------------------------------------------
\subsubsection{Random AP Mode Assignment with Power Control (Benchmark 2)}
%--------------------------------------------------
In this scheme, we assume that the AP modes ($\aaa$) are randomly assigned. Accordingly, we optimize the power control coefficients ($\ETAI$, $\ETAE$), under the same SE requirement constraints for IU and energy requirements for EUs. Therefore, the sum-SE and EE optimization problems are reduced to
%--------------------------------------------------------------------
\vspace{-0em}
\begin{subequations}\label{P:SE:Bch1}
	\begin{align}
		\text{\textbf{(P5)}:}~\underset{\ETAI, \ETAE}{\max}\,\, \hspace{2em}&
		\sum\nolimits_{k\in \Kd} \mathrm{SE}_k^{\Sto} (\ETAI, \ETAE)  
		\\
		\mathrm{s.t.} \,\,
		\hspace{2em}& Q_{\ell}^{\Sto} ( \ETAI, \ETAE)\geq \HEQoS(\tilgmal),~\forall \ell\in\mathcal{L},\\
		& \mathrm{SE}_k^{\Sto} (\ETAI, \ETAE)  \geq \SEQoS,~\forall k \in \Kd,\label{eq:infopower:ct2}\\
			&\sum\nolimits_{k\in\Kd}
        {\etamkI}\leq a_m,
        ~\forall m\in\MM,\label{eq:infopower:ct3}\\
              &\sum\nolimits_{\ell\in\mathcal{L}}
        {\etamlE}\leq 1-a_m,~\forall m\in\MM,\label{eq:infopower:ct4}
		\end{align}
\end{subequations}
%--------------------------------------------------------------------
and
%--------------------------------------------------------------------
\vspace{-0.4em}
\begin{subequations}\label{P:EEmaximiz:Bch1}
	\begin{align}
		\text{\textbf{(P6):}}~\underset{\ETAI, \ETAE}{\max}\,\, \hspace{1em}&~
		\EE^{\Sto} (\ETAI, \ETAE)  
		\\
		\mathrm{s.t.} \,\,
		\hspace{2em}&
   Q_{\ell}^{\Sto} ( \ETAI, \ETAE)\geq \HEQoS(\tilgmal),~\forall \ell\in\mathcal{L},\label{P:EEmaximiz:Bch1:c1}\\ 		&~\eqref{eq:infopower:ct2}-\eqref{eq:infopower:ct4},
		\end{align}
\end{subequations}
%-----------------------------------------------------------------
respectively, where $\SEk^{\Sto}( \ETAI, \ETAE)$, $Q_{\ell}^{\Sto} ( \ETAI, \ETAE)$, and $\EE^{\Sto} (\ETAI, \ETAE)$ are given in~\eqref{eq:SEk:Ex},~\eqref{eq:El_average:PPZF}, and~\eqref{eq:EE:def}, respectively, for a given $\qa$. The problems in~\eqref{P:SE:Bch1} and~\eqref{P:EEmaximiz:Bch1} have the same structure as~\eqref{P:SE} and~\eqref{P:EEmaximiz}, respectively. Therefore, we use the same solutions with some slight modifications. Moreover, the sum-HE maximization problem and the max-min average HE with fixed assignment vector $\qa$ follow the same structure as~\eqref{P:SHE:max} and~\eqref{P1:MaxMinEnergy}, respectively. 

\begin{remark}
  When the binary indicator $a_m$ is given and happens to be zero, we then set the corresponding $\etamkI$ to zero, excluding them from consideration as optimization variables in problems \eqref{P:SE:Bch1} and \eqref{P:EEmaximiz:Bch1}.
\end{remark}

%--------------------------------------------------
\subsubsection{SWIPT with Orthogonal Multiple Access (Benchmark 3)} 
%--------------------------------------------------
We assume that all APs are used for DL WIT and WPT over orthogonal time frames. More specifically, the duration of training phase is $\tau$ and the remaining duration $(\tau_c-\tau)$ is divided into two equal time fractions of length $(\tau_c-\tau)/2$. In the first fraction, all APs transmit DL information to the IUs by using the PZF precoding and in the consecutive fraction, all APs transmit energy symbols towards the EUs with  MRT precoding. Therefore, the average SE at the $k$-th IU and average HE at the $\ell$-th EU is given by
%-------------------
% \begin{figure*}
\begin{align}~\label{eq:SEk:S2}
    &\mathrm{SE}_k^{\Stre}(\ETAI)
      =
      \frac{1}{2}\Big(1\!- \!\frac{\tau}{\tau_c}\Big)
      \nonumber\\
      &
      \times\log_2
      \Bigg(
       1\! + \!\frac{
                  \rho \big(N-K_d\big)\Big(\sum_{m\in\MM}\sqrt{ \etamkI \gamuemk}  \Big)^2
                 }
                 { \sum_{m\in\MM}
                 \sum_{k'\in\Kd}
  \rho \etamkpI 
  \Big(\betamkue-\gamuemk\Big)
                   \!+\!  1}
     \Bigg),\hspace{2em},
\end{align}
 %  	\hrulefill
	% \vspace{-4mm}
 %  \end{figure*}
%-------------------
and
%---------------------------------------------------
\vspace{-0.5em}
\begin{align}~\label{eq:El:S2}
    Q_{\ell}^{\Stre} (\ETAE)
 =&
     \frac{(\tau_c-\tau)}{2}\Snn
     \bigg(\!
     {\rho}\big(N+\!1\big)\!\sum\nolimits_{m\in\MM}\!
     {\etamlE} \gamsml
     \nonumber\\     
     &+\!
     {\rho}\!\sum\nolimits_{m\in\MM}\sum\nolimits_{\ell'\in\LL\setminus \ell}\!
   {\etamlpE} \betamls \!+\! 1\!\bigg)\!.
\end{align}
%----------------------------------------------------

The total EE (bit/Joule) can be expressed as
%------------------
\begin{align}~\label{eq:EE:TDMA}
   \EE^{\Stre}( \ETAI)  = \frac{ B  \sum_{k\in\Kd}\mathrm{SE}_k^{\Stre} (\ETAI)}{P_\mathtt{total}^{\Stre} (\ETAI)},  
\end{align}
%----------------
where the total power consumption for information transmission is given by
%--------------------
\begin{align}~\label{eq:Ptotal:TDMA}
P_\mathtt{total}^{\Stre}(\ETAI)& =  \sum\nolimits_{m\in\mathcal{M}} \Big[\frac{1}{\zeta_m}\rho_{d}\Sn\left(\sum\nolimits_{k\in\Kd}  \etamkI\right)
\nonumber\\
&\hspace{-5em}
\!+ P_{\mathtt{bt},m} B \sum\nolimits_{k\in\Kd}\mathrm{SE}_k^{\Stre} (\ETAI)
+P_m^{\mathtt{fixed}} \Big]   + \sum\nolimits_{k\in\Kd} P_{\mathtt{D},k}.
\end{align}
%---------------------

\textbf{\textit{Sum-SE Maximization:}} By using~\eqref{eq:SEk:S2} and~\eqref{eq:El:S2}, the optimization problem in~\eqref{P:SE}  is reduced to the following power control problem
%--------------------------------------------------------------------
\begin{subequations}\label{P:SE:S2}
	\begin{align}
		\text{\textbf{(P7):}}~\underset{\ETAI, \ETAE}{\max}\,\, \hspace{2em}&
		\sum\nolimits_{k\in \Kd} \mathrm{SE}_k^{\Stre} (\ETAI) 
		\\
		\mathrm{s.t.} \,\,
		\hspace{2em}& Q_{\ell}^{\Stre}\big(\ETAE\big) \geq
		\HEQoS(\tilgmal),~\forall \ell\in \mathcal{L},\label{P:SE:S2:ct1}\\
		& \mathrm{SE}_k^{\Stre} (\ETA^I)  \geq \SEQoS,~\forall k\in\Kd,\\
			&\sum\nolimits_{k\in\Kd}
        {\etamkI}\leq 1,~\forall m\in\MM,\label{P:SE:S2:ct4}\\
       &\sum\nolimits_{\ell\in\mathcal{L}}
        {\etamlE}\leq 1,~\forall m\in\MM.\label{P:SE:S2:ct5}
		\end{align}
\end{subequations}
%--------------------------------------------------------------------

Since $\mathrm{SE}_k^{\mathtt{S}_2}(\ETAI)$ only depends on $\ETAI$ and $ Q_{\ell}^{\mathtt{S}_2} (\ETAE)$ is only a function of $\ETAE$, the power control problem can be decoupled into two separate problems. After applying SCA and some manipulations, the first sub-problem over $\ETAI$ is obtained as 
%----------------------------------
 \vspace{-0.2em}
\begin{subequations}\label{P:SHE:max:S2:R1}
	\begin{align}
		\text{\textbf{(P7-1a):}}~\underset{\ETAI, \qt^{\Stre}}{\max}\,\, \hspace{0.5em}&
		\prod\nolimits_{k\in \Kd} (1+t_k^{\Stre})
		\\
		\mathrm{s.t.} \,\,
		\hspace{0.5em}& t_{k}^{\Stre}\geq 2^{\tilSEQoS}-1,~\forall k\in\K,\\
           \hspace{-1em}& 
     \rho \big(N\!-\!K_d\big)
                  q_{k,\Stre}^{(n)}\bigg(\!2\!\!\sum_{m\in \MM}\!\!\sqrt{ \etamkI \gamuemk}\!-\! q_{k,\Stre}^{(n)} t_k^{\Stre} \bigg)
                  \nonumber\\
                  &\hspace{-0em}
    \geq \!\!\!\! \sum_{m\in\MM}
                 \sum_{k'\in\K}
  \rho \etamkpI 
  \Big(\betamkue\!-\!\gamuemk\Big)
                   \!+\!  1, ~\forall k\in\K,           
                     ~\label{eq:EH:ct2:S2:R1}\\
       &\sum\nolimits_{k\in\K}
        {\etamkI}\leq 1,~\forall m\in\MM,\label{eq:powerenergy:S2:R1a}
		\end{align}
\end{subequations}
%-----------------------------
where $\qt^{\Stre} =\{t_1^{\Stre},\ldots,t_{K_d}^{\Stre}\}$ are auxiliary variables $\tilSEQoS=\frac{2\SEQoS}{1- {\tau}/{\tau_c}}$ and $q_{k,\Stre}^{(n)}=\sum_{m=1}^{M}\sqrt{ \etamkIn \gamuemk}/t_k^{\Stre(n)}$.  This sub-problem can be iteratively solved using CVX~\cite{cvx}.

Moreover, the sub-problem over $\ETAE$ is the following feasibility problem, which can be efficiently solved using CVX~\cite{cvx}
%----------------------------------
\vspace{-1em}
\begin{subequations}\label{P:SHE:max:S2:R1}
	\begin{align}
		\text{\textbf{(P7-1b):}}~{\text{Find}}\,\, \hspace{0.5em}&
		\ETAE
		\\
		\mathrm{s.t.} \,\,
		\hspace{0.5em}&  
      %\frac{(\tau_c-\tau)}{2}\Snn
     %\bigg(\!
     1/ {\rho}\!+\!\sum\nolimits_{m\in\MM}\sum\nolimits_{\ell'\in\LL\setminus \ell}\!
   {\etamlpE} \betamls
   \!+\!
   \big(N\!+\!1\big)\!
        \nonumber\\     
     &\hspace{-3em}\times \sum\nolimits_{m\in\MM}\!
     {\etamlE} \gamsml\!\geq
		\frac{2\HEQoS(\tilgmal)}{(\tau_c-\tau)\Snn},~\forall \ell \in \LL,
  \label{P:SHE:max:P71a:ct1}\\
			&\sum\nolimits_{\ell\in\LL}
        {\etamlE}\leq 1,
        ~\forall m\in\MM.\label{eq:infopower:S2:R1b}
		\end{align}
\end{subequations}
%-----------------------------

\textbf{\textit{EE Maximization:}} For the considered benchmark, the optimization problem in~\eqref{P:EEmaximiz} is reduced to 
 %---------------------
\begin{subequations}\label{P:SHE:max:S2}
\begin{alignat}{2}
&\text{\textbf{(P8):}}~\underset{\ETAI, \ETAE}{\max}        
&\hspace{1em}& \EE^{\Stre} (\ETAI)\\
%----
&\hspace{2em}\text{s.t.} 
&         &  ~\eqref{P:SE:S2:ct1}-\eqref{P:SE:S2:ct5}.
\end{alignat}
\end{subequations}
%------------------------

%%%%%%%%%%%%%%%%%%%%%%%%%%%%%%%%%%%%%%%%%%%
\begin{table*}
	\centering
   	\caption{ Summary of optimization problems' parameters for computational complexity analysis}\label{tabel:complexity}
	\vspace{-0.6em}
	\small
\begin{tabular}{|p{1.8cm}|p{2.9cm}|p{2.5cm}|p{2.5cm}|}
	\hline
        \centering\textbf{Problem} 
        &\centering {$A_v$}
        &\centering  $A_l$
        &\centering  $A_q$
        \cr

        \hline

        \hspace{2em}\textbf{P1.3}     
        & \centering  $M(K_d+L+1) +K_d$
        & \centering $2M+K_d$ 
        & \centering  $M+K_d+L$
         \cr
        
        \hline

       \hspace{2em}\textbf{P2.2}         
        &\centering $M(K_d+L+1) +K_d$
        &\centering  $2M$
        &\centering  $M+K_d+L$
        \cr

           \hline
        \hspace{2em}\textbf{P3.2}        
        &\centering $M(K_d+L+1) +L$ 
        &\centering  $2M+L$
        & \centering $M+K_d+L$
        \cr

        \hline 

        \hspace{2em}\textbf{P4.2}     
        &\centering $M(K_d+L+1) +L$ 
        &\centering  $2M$
        &\centering   $M+K_d+L$
        \cr

        \hline
\end{tabular}
\label{Contribution}
\end{table*}
%%%%%%%%%%%%%%%%%%%%%%%%%%%%%%%%%%%%%%%%%%%%%%%%%%

%======================================
\begin{table*}
\vspace{0.5em}
	\caption{Parameters of power 
    consumption} 
	\vspace{-0.5em}
	\centering 
	\begin{tabular}{|c | c |}
		\hline
		\textbf{Parameter} & \textbf{Value}  \\ [0.5ex]
		\hline
Fixed power consumption/ each fronthaul  ($P_{\mathtt{fdl},m}$, $\forall m$)~\cite{Emil:TWC:2015:EE,ngo18TGN} & $0.825$ W   \\
		\hline
Internal power consumption/antenna  ($P_{\mathtt{cdl},m}, \forall m$)~\cite{ngo18TGN} & $0.2$ W   \\
		\hline
Traffic-dependent fronthaul power  ($P_{\mathtt{bt},m}, \forall m$)~\cite{Emil:TWC:2015:EE,ngo18TGN} & $0.25$ W/(Gbits/s) \\
		\hline
Power amplifier efficiency at the APs  ($\zeta_m$, $\forall m$)~\cite{ngo18TGN} & $0.4$   \\
		\hline
% Power amplifier efficiency at the UEs  ($\chi$)~\cite{bashar19TGCN} & $0.3$   \\
% 		\hline
Fixed power consumption to run the IU circuit  ($P_{\mathtt{D},k}, \forall k $)~\cite{bashar19TGCN} & $0.1$ W \\
\hline	
	\end{tabular}
	\label{tab:PowerconsumptionParameter}
	%\vspace{-2.1em}
\end{table*}
%=============================================

Now, by invoking~\eqref{eq:SEk:S2},~\eqref{eq:El:S2}, and~\eqref{eq:EE:TDMA}, the power control problem in~\eqref{P:SHE:max:S2} can be rewritten as
 %---------------------
  \vspace{-0.2em}
\begin{subequations}\label{P:EEmaximiz:P81a}
\begin{alignat}{2}
&\text{\textbf{(P8-1):}}\underset{\ETAI,\ETAE,\qv^{\Stre},\qt^{\Stre}}{\min}        
&\hspace{1em}&  \frac{\sum\nolimits_{m\in\mathcal{M}} \frac{1}{\zeta_m}\rho_{d}\Sn\left(\sum\nolimits_{k\in\Kd}  \etamkI\right)
        + k^{\Stre}
         }{B\sum_{k\in\K} v_k^{\Stre}}\\
%----
&\hspace{5em}\text{s.t.} 
&         &  \mathrm{SINR}_k^{\Stre} (\ETAI)  \geq t_k^{\Stre},~\forall k\in \Kd,\label{P:EEmaximiz:P81a:c2}\\
%----
&         &      & t_k^{\Stre} \geq 2^{2v_k^{\Stre}}-1
          ,~\forall k\in \Kd,\label{P:EEmaximiz:P81a:c3}\\
%----
&         &      &v_k^{\Stre}  \geq \SEQoS,~\forall k\in \Kd,\label{P:EEmaximiz:P81a:c4}\\
%----
&         &      &~\eqref{P:SE:S2:ct1},~\eqref{P:SE:S2:ct4},~\eqref{P:SE:S2:ct5},\label{P:EEmaximiz:P42:c6}
\end{alignat}
\end{subequations}
%------------------------
 where $k^{\Stre}\triangleq \sum\nolimits_{m\in\mathcal{M}} P_m^{\mathtt{fixed}}+\sum_{k\in\Kd} P_{\mathtt{D},k}$. Now, we take the natural logarithm from the linear fractional objective function and apply SCA, to deal with the non-convex constraint~\eqref{P:EEmaximiz:P81a:c2}, we get
%-------------------
 \vspace{-0.3em}
\begin{subequations}\label{P:EEmaximiz:P82a}
	\begin{align}
		\text{\textbf{(P8-2a):}}~\underset{\ETAI,\ETAE,\qv^{\Stre},\qt^{\Stre}}{\min}\,\, \hspace{0.05em}&~
  \log(g_1(\ETAI)) - \log(g_2(\qv^{\Stre}))
        %  \frac{\sum\nolimits_{m\in\mathcal{M}} \frac{1}{\zeta_m}\rho_{d}\Sn\left(\sum\nolimits_{k\in\Kd}  \etamkI\right)
        % + k^{\Stre}
        %  }{B\sum_{k\in\K} v_k^{\Stre}}
        		\label{P:EEmaximiz:P82a:ct1}\\
		\mathrm{s.t.} \,\,
		\hspace{1em}&
  ~~\eqref{eq:EH:ct2:S2:R1},\eqref{P:SHE:max:P71a:ct1},~\eqref{P:EEmaximiz:P81a:c3}-\eqref{P:EEmaximiz:P42:c6},
		\end{align}
\end{subequations}
 %-----------------
 where $g_1(\ETAI) = \sum\nolimits_{m\in\mathcal{M}} \frac{1}{\zeta_m}\rho_{d}\Sn\left(\sum\nolimits_{k\in\Kd}  \etamkI\right) + k^{\Stre}$ and
 $g_2(\qv^{\Stre}) = B\sum_{k\in\K} v_k^{\Stre}$. Since the first term in the objective function~\eqref{P:EEmaximiz:P82a:ct1} is concave, we can apply SCA to get the following convex optimization problem 
 %-------------------
 \vspace{-0.3em}
\begin{subequations}\label{P:EEmaximiz:P82a}
	\begin{align}
		\text{\textbf{(P8-2b):}}~\underset{\ETAI,\ETAE,\qv^{\Stre},\qt^{\Stre}}{\min}\,\, \hspace{0.05em}&~
  \log\big(g_1\big(\ETAI^{(n)}\big)\big) -\frac{g_1(\ETAI)}{g_1\big(\ETAI^{(n)}\big)}\nonumber\\
  &\hspace{2em}- \log\big(g_2(\qv^{\Stre})\big)
        %  \frac{\sum\nolimits_{m\in\mathcal{M}} \frac{1}{\zeta_m}\rho_{d}\Sn\left(\sum\nolimits_{k\in\Kd}  \etamkI\right)
        % + k^{\Stre}
        %  }{B\sum_{k\in\K} v_k^{\Stre}}
        		\label{P:EEmaximiz:P82a:ct1b}\\
		\mathrm{s.t.} \,\,
		\hspace{1em}&
  ~~\eqref{eq:EH:ct2:S2:R1},\eqref{P:SHE:max:P71a:ct1},~\eqref{P:EEmaximiz:P81a:c3}-\eqref{P:EEmaximiz:P42:c6}.
		\end{align}
\end{subequations}
 %-----------------

The solution to the \textbf{sum-HE maximization problem} for Benchmark $3$ has been provided in~\cite{Mohammadi:GC:2023}. Additionally, the \textbf{max-min average HE} problem can be efficiently solved using the SCA under the same benchmark.

%--------------------------------
\subsection{ Computational Complexity Analysis}
%---------------------------------
The relaxed optimization problems (\textbf{P1.3}), (\textbf{P2.2}), (\textbf{P3.2}), and (\textbf{P4.2}) each require a computational complexity of $\mathcal{O}(\sqrt{A_l + A_q}(A_v + A_l + A_q)A_v^2)$ per iteration~\cite{tam16TWC,Mohammadi:JSAC:2023}. Here $A_v$ represents the number of the real-valued scalar variables,  while $A_l$, and $A_q$ denote the number of linear and quadratic constraints, respectively. Table~\ref{tabel:complexity} summarizes the values of $A_v$, $A_l$, and $A_q$ for the various optimization problems. As demonstrated in Section~\ref{sec:num}, these problems converge to optimized solutions within only a few iterations.

%%%%%%%%%%%%%%%%%%%%%%%%%%%%%%%%%%%%%%%%%%%%%%%%%%%%%%%%%%%%%%%%%%%%%%%%%%
\vspace{2em}
\section{Numerical Examples}~\label{sec:num}
%%%%%%%%%%%%%%%%%%%%%%%%%%%%%%%%%%%%%%%%%%%%%%%%%%%%%%%%%%%%%%%%%%%%%%%%%%
%--------------------------------
\vspace{-2em}
\subsection{Large-scale Fading Model and System Parameters}
%--------------------------------
We assume that the $M$ APs, $K_d$ IUs, and $L$ EUs are uniformly distributed in a square of $1 \times 1$ km${}^2$, whose edges are wrapped around to avoid the boundary effects. The large-scale fading coefficients $\beta_{mk}\in\{\betamls,\betamkue\}$ are modeled as~\cite{emil20TWC}
%-----------------------------------------------------
\begin{align}\label{fading:large}
\beta_{mk} [\text{dB}] = -30.5-36.7\log_{10}\left(\frac{d_{mk}}{1\,\text{m}}\right) + F_{mk},
\end{align}
%-----------------------------------------------------
where $d_{mk}$ is the distance between UE $k$ and AP $m$ (computed as the minimum over different wrap-around cases)
and $F_{mk}\sim \mathcal{CN}(0,4^2)$ is the shadow fading. The shadowing terms from an AP to different UEs are correlated as
%-----------------------------------------------------
\begin{align}\label{corr:shadowing}
\mathbb{E}\{F_{mk}F_{jk'}\} \triangleq
\begin{cases}
 4^22^{-\delta_{kk'}/9\,\text{m}},& \text{if $j=m$},\\
 0, & \mbox{otherwise},
\end{cases}
\end{align}
%-----------------------------------------------------
where $\delta_{kk'}$ is the physical distance between UEs $k$ and $k'$.

The values of the network parameters are $\tau=200$, and $\tau_t=K_d+L$. We further set the bandwidth $B=50$ MHz
and noise figure $F = 9$ dB. Thus, the noise power $\Sn=k_B T_0 B F$, where $k_B=1.381\times 10^{-23}$ Joules/${}^o$K is the Boltzmann constant, while $T_0=290^o$K is the noise temperature. Let $\tilde{\rho}= 1$ W,  and $\tilde{\rho}_t = 0.25$~W be the maximum transmit power of the APs and uplink training pilot sequences, respectively. The normalized maximum transmit powers ${\rho}$ and ${\rho}_t$ are calculated by dividing these powers by the noise power. The non-linear EH parameters are set as $\xi=15\times 10^{3}$, $\chi=0.22\times 10^{-3}$, and $\phi=0.39$ mW~\cite{Xiong:TWC:2017}. 
Moreover, the power consumption parameters are taken from~\cite{Emil:TWC:2015:EE,ngo18TGN,bashar19TGCN}  and shown in Table~\ref{tab:PowerconsumptionParameter}. 

For a fair comparison, in both the proposed and benchmark schemes, we handle infeasibility in the same manner: if the problem becomes infeasible — either due to the inability to satisfy the minimum SE requirement of the $k$-th IU ($\SEQoS$) or the minimum harvested power requirement at the $\ell$-th EU ($\Gamma_{\ell}$) — the output sum-SE for that network realization is set to zero.

%==========================================
\begin{figure}[t]
\centering
\includegraphics[width=0.46\textwidth]{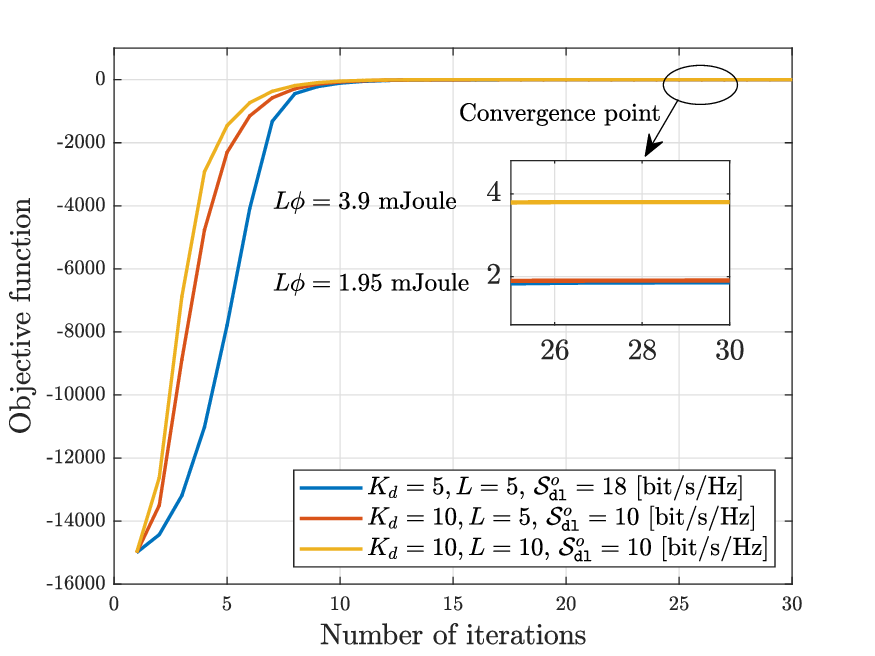}
\caption{ Convergence behavior of the optimization problem (\textbf{P3.2}) ($M=60$, $N=12$, $\Gamma_{\ell}=250$ $\mu$Joule).}
\vspace{-1em}
\label{figconv}
\end{figure}
%==========================================

%==========================================
\begin{figure}[t]
\centering
\includegraphics[width=0.46\textwidth]{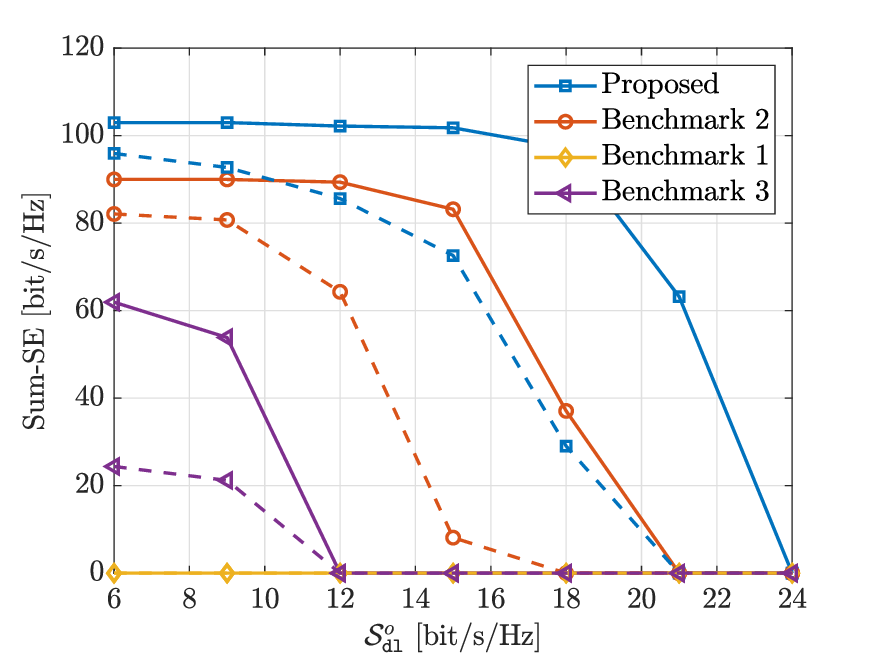}
%\vspace{-1.2em}
\caption{Average sum-SE versus the per-IU SE requirement. The solid lines depict results for $\tilde{\rho}_t = 0.25$~W, while the dashed lines show results for $\tilde{\rho}_t = 0.05$~W ($M=40$, $N=12$, $K_d=L=5$, $\Gamma_{\ell}=250$ $\mu$Joule).}
\vspace{0em}
\label{figsumSEsdl}
\end{figure}
%==========================================

%--------------------------------
\vspace{-1em}
\subsection{Results and Discussions }
%--------------------------------

%%%%%%%%%%%%%%%%%%%%%%%%%%%%%%%%%%%%%%%%%%%%%%%%%%%%%%%%%%%%%%%%%%%%%%%%%%%%%%%%%%%
\subsubsection{  Convergence behavior analysis}
%%%%%%%%%%%%%%%%%%%%%%%%%%%%%%%%%%%%%%%%%%%%%%%%%%%%%%%%%%%%%%%%%%%%%%%%
In Fig.~\ref{figconv}, we compare the convergence rate of the optimization problem (\textbf{P3.2}) for different number of IUs and EUs in the network and for a convergence percision of $\epsilon=10^{-5}$. To solve~\eqref{P:SHE:max3}, we use the convex conic solver MOSEK and set $c_2=1500$.   By increasing the number of IUs from $5$ to $10$, we reduced $\SEQoS$ from $18$ [bit/s/Hz] to $10$ [bit/s/Hz], to ensure the feasibility of the optimization problem.  When the optimization problem begins to converge, i.e., $a_m\approx a_m^{(n)}\in\{0,1\}$, then the value of the objective function converges to the maximum level of sum-HE, denoted by $L\phi$. We can see that with a small number of iterations (less than $30$ iterations), problem (\textbf{P3.2}) returns the optimized solution.  Furthermore, it is worth mentioning that the resulting values of the parameters $a_m$ converge to $1$ and $0$ with high accuracy. The other optimization problem discussed in the paper also exhibits similar convergence behavior, but it is not shown here due to space constraints and to avoid repetition.

%%%%%%%%%%%%%%%%%%%%%%%%%%%%%%%%%%%%%%%%%%%%%%%%%%%%%%%%%%%%%%%%%%%%%%%%%%%%%%%%%%%
\subsubsection{Effectiveness of the proposed scheme in terms of sum-SE}
%%%%%%%%%%%%%%%%%%%%%%%%%%%%%%%%%%%%%%%%%%%%%%%%%%%%%%%%%%%%%%%%%%%%%%%%%%%%%%%%%%%
Figure~\ref{figsumSEsdl} shows the average sum-SE versus the per-IU SE requirement and for two different levels of  pilot power. This figure provides insights into the maximum achievable $\SEQoS$ levels for different schemes and specific network setups. Specifically, we observe that SWIPT through orthogonal multiple access (Benchmark $3$) fails to support predefined SE requirements greater than $12$ \bsHz~ for IUs and the predefined average HE of $\Gamma_{\ell}=250$ $\mu$Joule for EUs. Moreover, with random AP mode selection and equal power allocation (Benchmark $1$), neither the IUs nor the EUs' predefined requirements are met even for small values of $\SEQoS$. However, by applying power control design, we can extend the coverage for IUs and EUs to $\SEQoS = 15$ \bsHz. Furthermore, with optimal AP mode selection and power control, this coverage is extended to $\SEQoS = 21$ \bsHz, and the overall sum-SE shows a significant improvement over Benchmark $2$, especially in the high $\SEQoS$ regime.  Finally, we observe that reducing the pilot power from $0.25$ W to $0.05$ W leads to a significant decrease in the achievable sum-SE of the system. This decline is attributed to the reduced channel estimation accuracy and the increased channel estimation error.

%==========================================
\begin{figure}[t]
\centering
\includegraphics[width=0.46\textwidth]{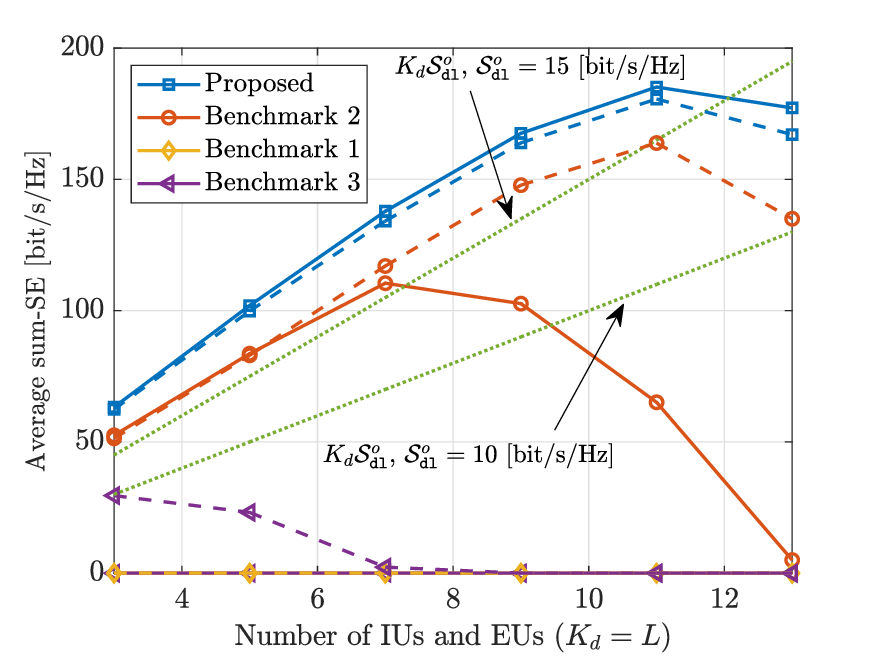}
%\vspace{-1.2em}
\caption{Average sum-SE versus the number of IUs and EUs with $K_d=L$. The dashed lines depict results for $\SEQoS=10$ \bsHz, while the solid lines show results for $\SEQoS=15$ \bsHz. ($M=40$, $N=12$, $\Gamma_{\ell}=250$ $\mu$Joule).}
\vspace{0em}
\label{figSumSEKu}
\end{figure}
%==========================================

Figure~\ref{figSumSEKu} shows the average sum-SE versus the number of IUs and EUs in the network for different system designs. Our numerical results lead to the following conclusions:
\begin{itemize}
    \item  When $\SEQoS = 15$ bit/s/Hz, only our proposed design and Benchmark $2$ meet the SE requirements for both IUs and EUs. However, as the number of IUs and EUs exceeds $K_d + L = 18$, Benchmark $2$ fails to satisfy the SE requirement for some network realizations. Specifically, our analysis (not shown in this figure) reveals that Benchmark $2$ fails in $35\%$ of network realizations. This is due to the average sum-SE of Benchmark $2$ falling below the minimum acceptable level. By reducing $\SEQoS$ to $10$ bit/s/Hz, both the proposed scheme and Benchmark $2$ can support up to $K_d + L = 26$ users with $MN = 480$ antenna services. To accommodate a greater number of users, time division multiple access (TDMA) or frequency division multiple access (FDMA) should be employed, or the number of antenna services in the network  must be scaled up (i.e., increasing both $M$ and $N$).\footnote{ Note in CF-mMIMO networks aggregate number of antennas across all APs is significantly larger than the number of users jointly scheduled on the same time-frequency resource block~\cite{Ngo:PROC:2024}. Unlike BSs in conventional cellular networks, APs typically have much lower power consumption and reduced signal processing burdens, as the CPU handles the more computationally intensive signal processing tasks. This allows for the deployment of a large number of APs in such networks~\cite{Mohammadi:PROC.2024}}
    \item By reducing $\SEQoS$ from $15$ \bsHz~to $10$ \bsHz, we observe that Benchmark $3$ can partially support the requirements of IUs and EUs, when their numbers are small. However, the proposed design and Benchmark $2$ still provide significant performance gains over Benchmark $3$. This result indicates that our proposed design can support dense CF-mMIMO networks with medium to large SE requirements for IUs.
    \item As long as the proposed design and Benchmark $2$ guarantee the predefined requirements for IUs and EUs, they achieve almost the same maximum average sum-SE. This observation is consistent with the results shown in Fig.~\ref{figsumSEsdl}.   
    %\item \com{will be completed}. 
\end{itemize}

%==========================================
\begin{figure}[t]
\centering
\includegraphics[width=0.46\textwidth]{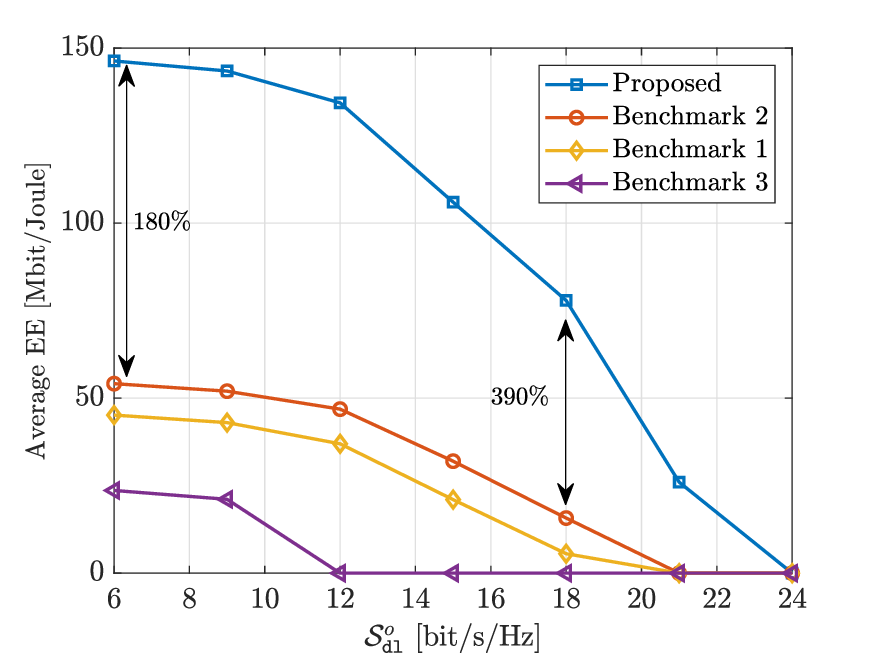}
\vspace{-0.5em}
\caption{Average EE versus the per-IU SE requirement ($M=40$, $N=12$, $K_d=L=5$, $\Gamma_{\ell}=250$ $\mu$Joule).}
\vspace{0em}
\label{figEEsdl}
\end{figure}
%==========================================

%%%%%%%%%%%%%%%%%%%%%%%%%%%%%%%%%%%%%%%%%%%%%%%%%%%%%%%%%%%%%%%%%%%%%%%%%%%%%%%%%%%
\subsubsection{Effectiveness of the proposed scheme in terms of EE}
%%%%%%%%%%%%%%%%%%%%%%%%%%%%%%%%%%%%%%%%%%%%%%%%%%%%%%%%%%%%%%%%%%%%%%%%%%%%%%%%%%%
Figure~\ref{figEEsdl} shows the average EE of various schemes in relation to the per-IU SE requirement. Our proposed design demonstrates substantial gains over all benchmarks. Specifically, it achieves up to $180\%$ and $390\%$ improvements over Benchmark $2$ at $\SEQoS = 6$ and $18$ \bsHz, respectively. These significant gains underscore the critical role of AP mode selection. This improvement is not only due to increased SE, but also to the significant reduction in the transmit power at the APs to meet the requirements of both the IUs and EUs. Efficient management of APs, specifically utilizing nearby APs to serve the EUs, can substantially reduce the required power levels at the APs and minimize the interference for the IUs. Another observation is that the commonly used structure of Benchmark $3$ in the literature, designed to support SWIPT in CF-mMIMO networks, yields significantly poorer EE performance compared to our proposed design. 

%==========================================
\begin{figure}[t]
\centering
\includegraphics[width=0.46\textwidth]{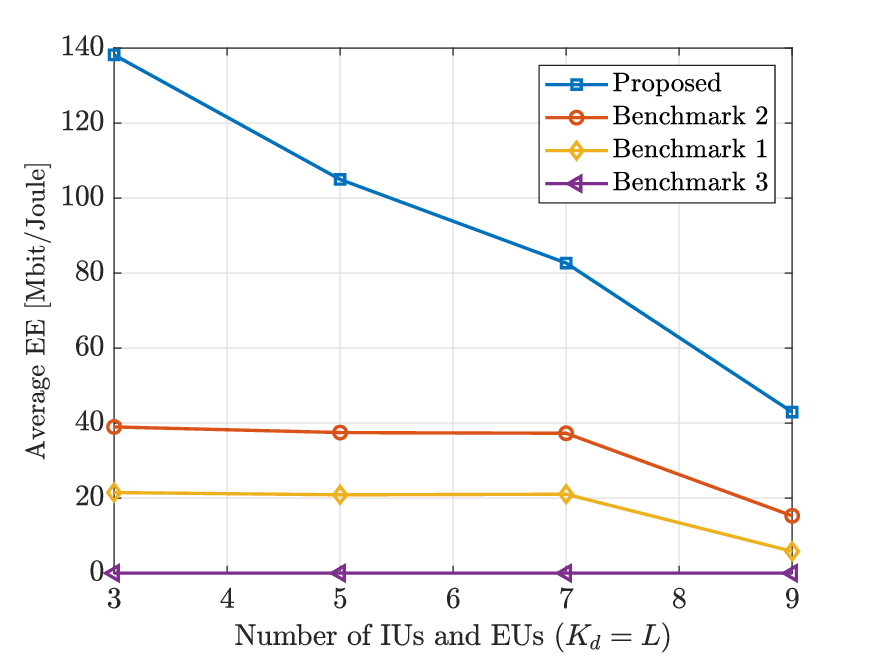}
\vspace{-0.5em}
\caption{Average EE versus the number of IUs and EUs ($M=40$, $N=12$, $K_d=L=5$, $\SEQoS=15$ \bsHz, $\Gamma_{\ell}=250$ $\mu$Joule).}
\vspace{0em}
\label{figEEvKu}
\end{figure}
%==========================================

In Fig.~\ref{figEEvKu}, we compare the average EE of different designs against different numbers of IUs and EUs with $K_d = L$. Our numerical results lead to the following conclusions:
\begin{itemize}
    \item
    For a given number of APs, $M = 40$, increasing the number of IUs and EUs raises the power consumption levels in the network to meet their requirements. Additionally, the likelihood that the predefined requirements of some IUs and EUs cannot be met increases. Consequently, a drop in EE is observed. 
    \item As the number of IUs and EUs increases from $3$ to $7$, the average EE of Benchmark $1$ and Benchmark $2$ shows a slight decrease. In contrast, the average EE of our proposed design decreases from $140$ [Mbit/Joule] to $80$ [Mbit/Joule]. This behavior can be attributed to the significant variation in the power consumption at the APs, which is needed to meet the SE and HE requirements for different numbers of IUs and EUs. When the number of IUs and EUs is small, their requirements can be met with relatively low transmit power at the APs. However, as the number of IUs and EUs increases, higher transmit power levels are necessary to meet their requirements. However, in Benchmark $1$ and Benchmark $2$, there is no significant difference in the transmit power levels required for varying numbers of UEs.
\end{itemize}

%==========================================
\begin{figure}[t]
\centering
\includegraphics[width=0.46\textwidth]{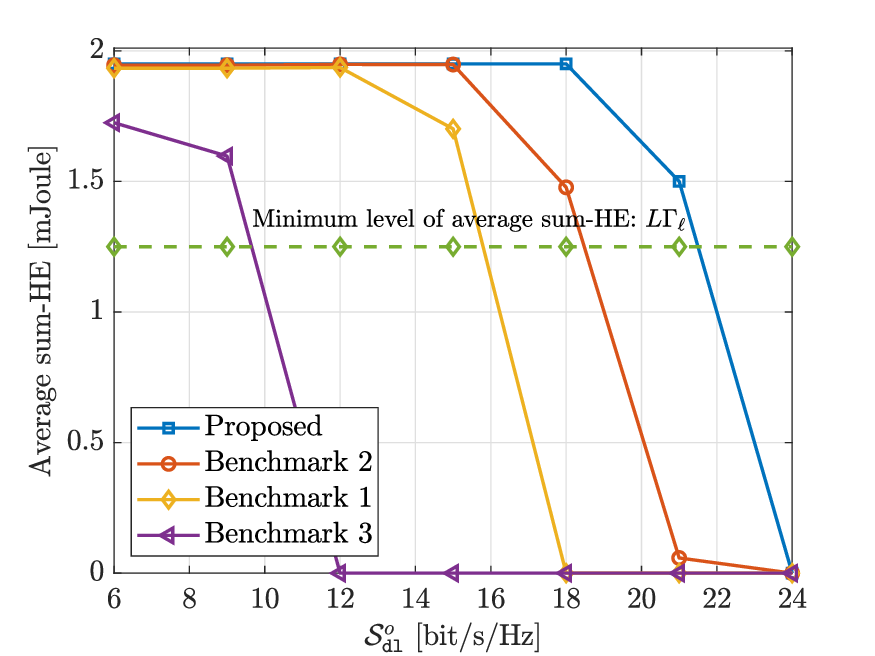}
\vspace{-0.5em}
\caption{Average sum-HE versus the per-IU SE requirement ($M=40$, $N=12$, $K_d=L=5$, $\Gamma_{\ell}=250$ $\mu$Joule).}
\vspace{0em}
\label{figsumHEsdl}
\end{figure}
%==========================================

%==========================================
\begin{figure}[t]
\centering
\includegraphics[width=0.46\textwidth]{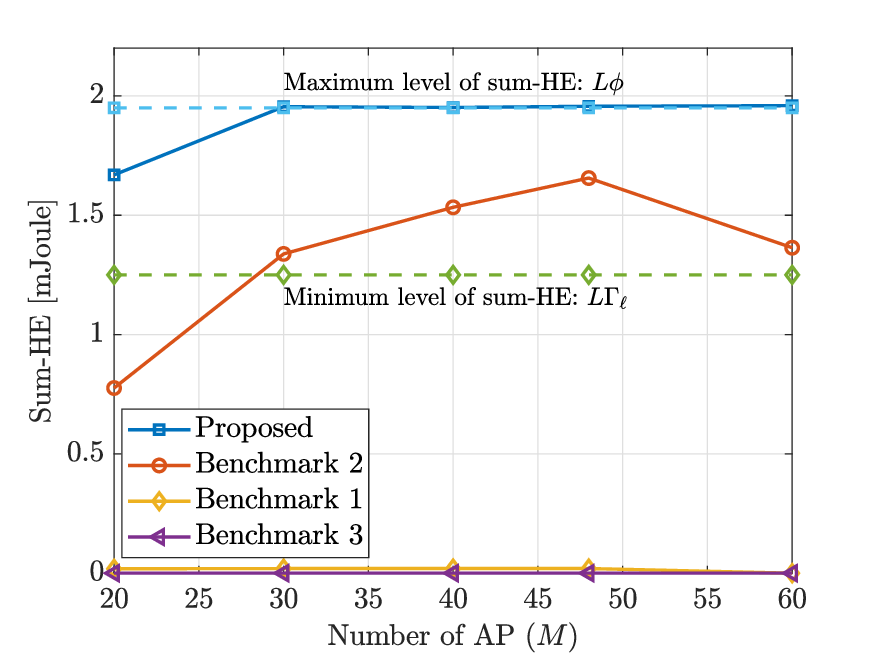}
\vspace{-0.5em}
\caption{Average sum-HE versus the number of APs, $M$, ($MN=480$, $K_d=L=5$, $\SEQoS=18$ \bsHz$, \Gamma_{\ell}=250$ $\mu$Joule).}
\label{figsumHEM}
\end{figure}
%==========================================

%%%%%%%%%%%%%%%%%%%%%%%%%%%%%%%%%%%%%%%%%%%%%%%%%%%%%%%%%%%%%%%%%%%%%%%%%%%%%%%%%%%
\subsubsection{Effectiveness of the proposed scheme in terms of sum-HE}
%%%%%%%%%%%%%%%%%%%%%%%%%%%%%%%%%%%%%%%%%%%%%%%%%%%%%%%%%%%%%%%%%%%%%%%%%%%%%%%%%%%
Figure~\ref{figsumHEsdl} shows the average sum-HE as a function of $\SEQoS$, for $MN=480$ antenna services and $\Gamma_{\ell}=250$ $\mu$Joule. We can observe that Benchmark $3$ meets moderate SE requirements (less than $9$~\bsHz) in the network while satisfying the predefined HE requirements at the EUs. However, the CF-mMIMO system, with simultaneous energy and information transmission across the entire coherence interval, significantly expands the rate-energy region. Specifically, Benchmark $1$ supports $\SEQoS$ up to $15$~\bsHz, whereas Benchmark $2$ and our proposed scheme achieve even higher rates of up to $18$~\bsHz~and $21$~\bsHz, respectively.

In Fig.~\ref{figsumHEM}, we examine the impact of the number of APs on the average sum-HE performance of different designs. The number of antenna services, denoted as $MN$, remains constant in this figure. Therefore, increasing $M$ results in fewer antennas per AP. It is evident that our proposed design effectively meets the requirements of both IUs and EUs across varying numbers of APs, with performance improving as the network becomes more distributed. However, Benchmark $2$ does not meet these requirements when the number of APs or the number of antennas per AP is small. Additionally, both Benchmark $1$ and Benchmark $3$ fail to meet the SE and HE requirements of both IUs and EUs.

%==========================================
\begin{figure}[t]
\centering
\includegraphics[width=0.47\textwidth]{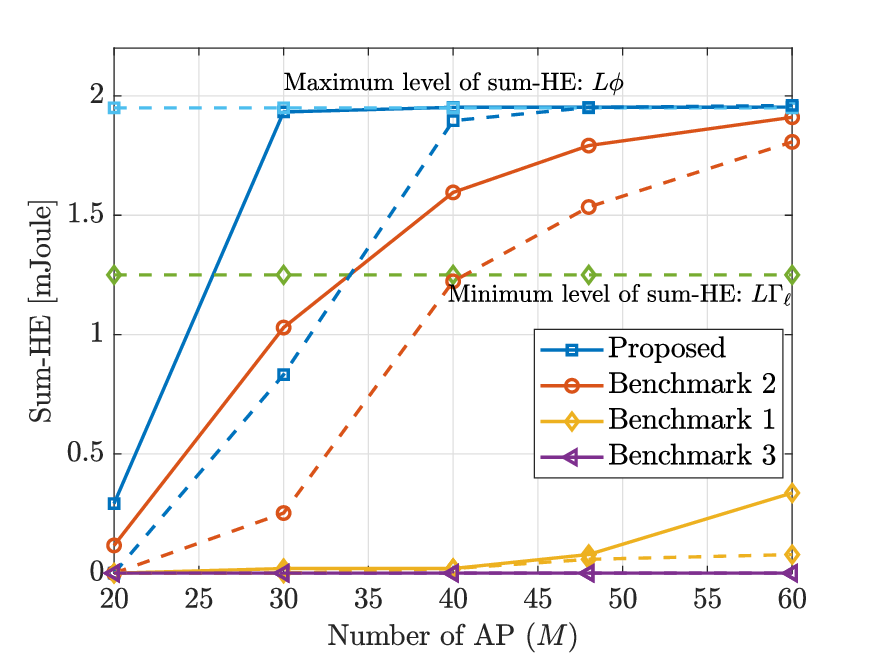}
\vspace{-0.3em}
\caption{Average sum-HE versus the number of APs, $M$, for a fixed number of antennas per each AP and different levels of transmit power at the APs. The solid lines depict results for $\tilde{\rho}= 1$ W, while the dashed lines show results for $\tilde{\rho}= 0.5$ W ($N=12$, $K_d=L=5$, $\SEQoS=18$ \bsHz, $\Gamma_{\ell}=250$ $\mu$Joule).}
\vspace{-0.2em}
\label{figsumHEMNfixed}
\end{figure}
%==========================================

Figure~\ref{figsumHEMNfixed} illustrates the sum-HE as a function of the number of APs, $M$, with $N=12$  antennas per AP. The results are presented for two different transmit power levels at the APs:  $\tilde{\rho}= 1$ W and $\tilde{\rho}= 0.5$. As the number of APs increases, the distance between the E-APs and EUs decreases, thus ensuring that the required level of HE at the EUs is met by the proposed algorithm. Simultaneously, the QoS requirement for the IUs, i.e., 
$\SEQoS=18$ \bsHz, is maintained. While Benchmark 2 can meet both the EU and IU requirements for a higher numbers of APs, the other two benchmarks fail to satisfy the network demands. Moreover, by reducing the transmit power (and consequently the downlink SNR) by a factor of 2, the proposed algorithm requires double the number of APs (i.e., $M=40$) to achieve the maximum level of sum-HE.

%%%%%%%%%%%%%%%%%%%%%%%%%%%%%%%%%%%%%%%%%%%%%%%
\vspace{-1em}
\section{Conclusion}~\label{sec:conc}
%%%%%%%%%%%%%%%%%%%%%%%%%%%%%%%%%%%%%%%%%%%%%%%
We investigated the problem of sum-SE and EE maximization in CF-mMIMO systems with separate IUs and EUs. We proposed a novel architecture, where, based on the network requirements, the AP operation mode and the associated  power control coefficients were jointly  optimized. Furthermore, we extended our optimization framework to address sum-HE maximization and max-min HE optimization problems in the considered CF-mMIMO network. We showed that our proposed design significantly enhances the EE, SE, and EH efficiency, especially in dense CF-mMIMO networks with high SE and HE requirements for IUs and EUs, respectively. Our results also confirm that, even with random AP mode selection and power control, substantial performance gains are achieved over CF-mMIMO counterparts that perform WPT and WIT over orthogonal coherence intervals. Moreover, given a fixed number of service antennas, there is an optimal configuration for the number of APs and per-AP antennas in both the proposed scheme and the random AP mode selection scheme with power control, which maximizes the average sum-HE.

A potential direction for future research is the development of fast-converging algorithms to improve the efficiency of the proposed optimization methods by reducing the run-time per iteration. This is particularly crucial in highly dynamic scenarios where optimization variables must be updated in real-time.

%%%%%%%%%%%%%%%%%%%%%%%%%%%%%%%%%%%%%%%%%%%%%%%%%%%%%%%
\vspace{-1em}
\appendices
%%%%%%%%%%%%%%%%%%%%%%%%%%%%%%%%%%%%%%%%%%%%%%%%%%%%%%%

%%%%%%%%%%%%%%%%%%%%%%%%%%%%%%%%%%%%%%%%%%%%%%%%%%%%%%%
\section{Proof of Theorem~\ref{Theorem:SE:PPZF}}
\label{Theorem:SE:PPZF:proof}
%%%%%%%%%%%%%%%%%%%%%%%%%%%%%%%%%%%%%%%%%%%%%%%%%%%%%%

To derive $\SINRk$, by using~\eqref{eq:DS},  we first obtain $\mathrm{DS}_k$ as
% \begin{itemize}
%     \item Desired signal ($\mathrm{DS}_k$)
%----------------------
\begin{align}\label{eq:derivations}
 \mathrm{DS}_k  &=  \sum\nolimits_{m\in\MM}\sqrt{\rho a_m\etamkI} \Ex\Big\{\big(\gmkue\big)^H\wimk^\PZF \Big\}
  \nonumber\\
 % &=\sum\nolimits_{m\in\MM}\sqrt{\rho a_m\etamkI} \Ex\Big\{\Big(\big(\ghmkue\big)^H \!+\! \big(\gtilmkue\big)^H\Big)\wimk^\PZF \Big\}
 %  \nonumber\\
  &=\sum\nolimits_{m\in\MM}\sqrt{\rho a_m\etamkI} \Ex\Big\{ \big(\ghmkue\big)^H \wimk^\PZF \Big\}
   \nonumber\\
    &=\sum\nolimits_{m\in\MM}\sqrt{\rho \big(N-K_d\big) a_m\etamkI \gamuemk}.
\end{align}
%-----------------------

The beamforming gain uncertainty, $\mathrm{BU}_k$, can be derived as
%-----------------------
\begin{align}\label{eq:derivations:BU}
\Ex\Big\{ \big\vert  \mathrm{BU}_k  \big\vert^2\Big\} &\!=\! 
\Ex\Big\{ \Big\vert  \! \sum\nolimits_{m\in\MM}\!\!
 \Big(\sqrt{\rho a_m\etamkI}\big(\gmkue\big)^H\wimk^\PZF  \nonumber\\
 &\hspace{1em}-\sqrt{\rho a_m\etamkI} \Ex\Big\{  \big(\gmkue\big)^H\wimk^\PZF \Big\}  \Big)    \Big\vert^2\Big\} 
 \nonumber\\
%  & = 
% \Ex\Big\{ \Big\vert   \sum_{m\in\MM}
%  \sqrt{\rho a_m\etamkI}\big(\gmkue\big)^H\wimk^\PZF \Big\vert^2\Big\}
% - \big\vert\mathrm{DS}_k \big\vert^2
%  \nonumber\\
%  &=
%  \Ex\Big\{ \Big\vert   \sum_{m\in\MM}
%  \sqrt{\rho a_m\etamkI}\big(\ghmkue + \gtilmkue\big)^H\wimk^\PZF \Big\vert^2\Big\}
%  - \big\vert\mathrm{DS}_k \big\vert^2
%  \nonumber\\
%   &=\!
% \sum\nolimits_{m\in\MM}\!
%  {\rho a_m\etamkI} \Ex\Big\{ \Big\vert   \big(\gtilmkue\big)^H\wimk^\PZF \Big\vert^2\Big\}
%  \nonumber\\
   &=
\sum\nolimits_{m\in\MM}
 {\rho a_m\etamkI} \Big(\betamkue-\gamuemk\Big).
\end{align}
%-----------------------
By invoking~\eqref{eq:IUI}, the interference caused by the $k'$-th IU, is given by
%----------------------
\begin{align}\label{eq:derivations:IUI}
 \Ex\Big\{ \big\vert\mathrm{IUI}_{kk'} \big\vert^2\Big\}&\!=\!   \Ex\Big\{\big\vert\!\!\sum\nolimits_{m\in\MM}\!\!
  \sqrt{\rho a_m\etamkpI}\big(\gmkue\big)^H\wimkp^{\PZF} \Big\vert^2\Big\} 
 \nonumber\\
 % &=
 %  \Ex\Big\{\Big\vert\sum_{m\in\MM}
 %  \sqrt{\rho a_m\etamkpI}\big((\ghmkue)^H + (\gtilmkue)^H\big)\wimkp^{\PZF} \Big\vert^2\Big\},
 %  \nonumber\\
  &\hspace{-1em}\!=\!
   \Ex\Big\{\Big\vert\sum\nolimits_{m\in\MM}\!\!
  \sqrt{\rho a_m\etamkpI}\big(\ghmkue\big)^H \wimkp^{\PZF} \Big\vert^2\Big\}
  \nonumber\\
  &\hspace{-1em}+
   \Ex\Big\{\Big\vert\sum\nolimits_{m\in\MM}
  \sqrt{\rho a_m\etamkpI} \big(\gtilmkue\big)^H\wimkp^{\PZF} \Big\vert^2\Big\}
  \nonumber\\
  % &\hspace{-1em}=   \sum\nolimits_{m\in\MM}\sum\nolimits_{n\in\MM}
  % \sqrt{\rho a_m\etamkpI a_{n}\eta_{nk'}} \nonumber\\
  % &\times\Ex\Big\{\big(\wimkp^{\PZF}\big)^H \gtilmkue \big(\gtilnkue\big)^H \winkp^\PZF \Big\},
  % \nonumber\\
    &\hspace{-1em}=   \sum\nolimits_{m\in\MM}
  \rho a_m\etamkpI  
  \Big(\betamkue-\gamuemk\Big).
\end{align}
%-----------------------

The interference caused by the $\ell$-th EU, can be obtained as
%--------------------------
\begin{align}\label{eq:derivations:EUI}
 \Ex \big\{ \big\vert  \mathrm{EUI}_{k\ell} \big\vert^2\big\}  &
 \nonumber\\
 &\hspace{-5em}=
 \Ex \Big\{ \Big\vert  \sum\nolimits_{m\in\MM}
   \sqrt{\rho (1-a_m)\etamlE}
   \big(\gmkue\big)^H\weml^{\PMRT}  \Big\vert^2\Big\} 
   \nonumber\\
   %  &\hspace{-5em}= 
   % \sum\nolimits_{m\in\MM}
   % {\rho (1-a_m)\etamlE}
   % \Ex \Big\{ \Big\vert\big(\gmkue\big)^H\weml^{\PMRT}  \Big\vert^2\Big\} 
   % \nonumber\\
   %     &\hspace{-5em}= 
   %        \sum_{m\in\MM}
   % {\rho (1-a_m)\etamlE}
   % \Ex \Big\{ \Big\vert\big(\ghmkue + \gtilmkue\big)^H\weml^{\PMRT}  \Big\vert^2\Big\} 
   % \nonumber\\
       &\hspace{-5em}=
       \sum\nolimits_{m\in\MM}
   {\rho (1-a_m)\etamlE}
   \Ex \Big\{ \Big\vert\big(\gtilmkue\big)^H\weml^{\PMRT}  \Big\vert^2\Big\} 
   \nonumber\\
    &\hspace{-5em}=
       \sum\nolimits_{m\in\MM}
   {\rho (1-a_m)\etamlE}
      \Big(\betamkue-\gamuemk\Big).
  \end{align}
%--------------------------

To this end, by substituting~\eqref{eq:derivations},~\eqref{eq:derivations:BU},~\eqref{eq:derivations:IUI}, and~\eqref{eq:derivations:EUI} into~\eqref{eq:SINE:general}, the desired result in~\eqref{eq:SINE:PPZF} is obtained.

% %%%%%%$$$$$$$$$$$$$$$$$$$$$$$$$$$
\balance
\bibliographystyle{IEEEtran}
\bibliography{IEEEabrv,references}
% %%%%%%$$$$$$$$$$$$$$$$$$$$$$$$$$$

\begin{IEEEbiography}[{\includegraphics[width=1in,height=1.25in,clip,keepaspectratio]
{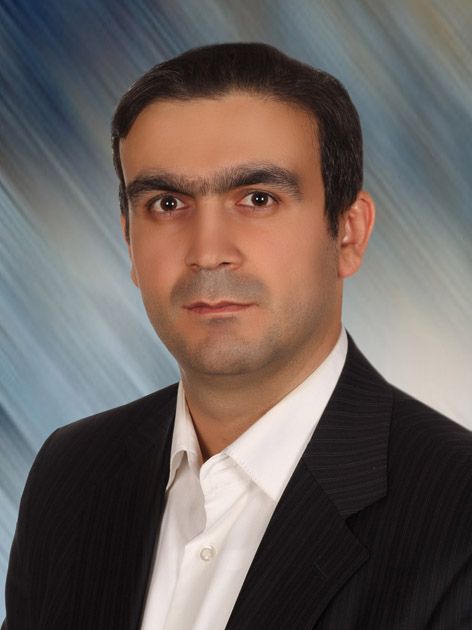}}]
{Mohammadali Mohammadi} (Senior Member, IEEE) is currently a Lecturer at the Centre for Wireless Innovation (CWI), Queen’s University Belfast, U.K. He previously held the position of Research Fellow at CWI from 2021 to 2024. His research interests include signal processing for wireless communications, cell-free massive MIMO, wireless power transfer, OTFS modulation, reconfigurable intelligent surfaces, and full-duplex communication. He has published more than 80 research papers in accredited international peer reviewed journals and conferences in the area of wireless communication and has co-authored two invited book chapters. He serves as an Associate Editor for IEEE Communications Letters and IEEE Open Journal of the Communications Society. He was a recipient of the Exemplary Reviewer Award for IEEE Transactions on Communications in 2020 and 2022, and IEEE Communications Letters in 2023. He has been a member of Technical Program Committees for many IEEE conferences, such as ICC, GLOBECOM, and VTC.
\end{IEEEbiography}

\vspace{-2em}

\begin{IEEEbiography}[{\includegraphics[width=1in,height=1.25in,clip,keepaspectratio]
{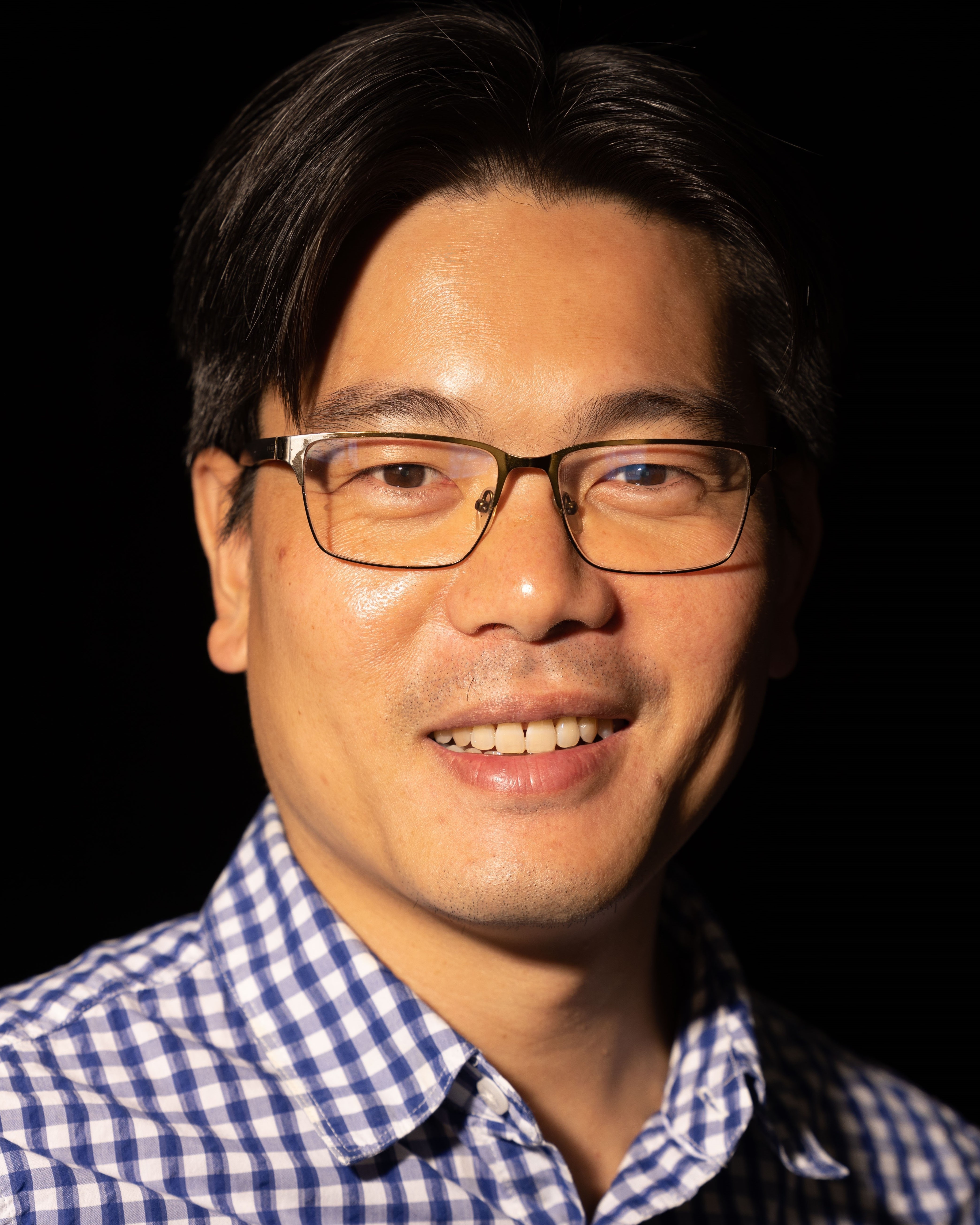}}]
{Le-Nam Tran} (Senior Member, IEEE) received his Ph.D. degree in radio engineering from Kyung Hee University, South Korea, in 2009. He is currently an Associate Professor with the School of Electrical and Electronic Engineering, University College Dublin, Ireland. His research interests include applications of optimization techniques for wireless communications design. Some recent topics include energy-efficient communications, physical layer security, cloud radio access networks, cell-free massive MIMO, and reconfigurable intelligent surfaces.
He was a recipient of the Career Development Award from Science Foundation Ireland in 2018. He was a co-recipient of the IEEE GLOBECOM 2021 Best Paper Award and the IEEE PIMRC 2020 Best Student Experimental Paper Award. He is an Associate Editor of IEEE Transactions on Communications, IEEE Communications Letters, and EURASIP Journal on Wireless Communications and Networking.  He has also served on the technical program committees for several major IEEE conferences. 

\end{IEEEbiography}

\vspace{-2em}

\begin{IEEEbiography}[{\includegraphics[width=1in,height=1.25in,clip,keepaspectratio]{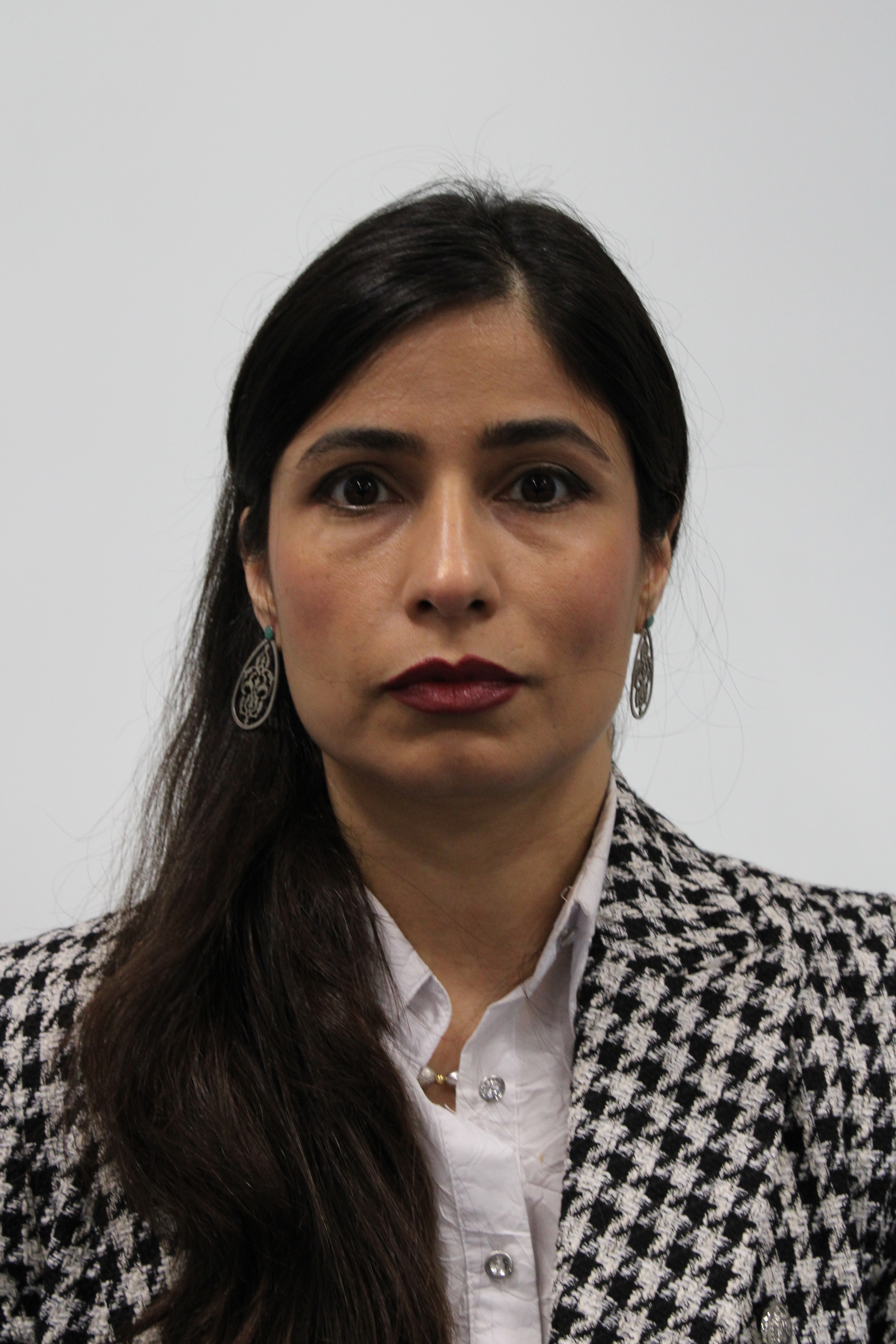}}]{Zahra Mobini}
 received the B.S. degree in electrical engineering from Isfahan University of Technology, Isfahan, Iran, in 2006, and the M.S and Ph.D.
degrees, both in electrical engineering, from the M. A. University of Technology and K. N. Toosi University of Technology, Tehran, Iran, respectively. From November 2010 to November 2011, she was a Visiting Researcher at the Research School of Engineering, Australian National University, Canberra, ACT, Australia. She is currently  a Post-Doctoral Research Fellow at the Centre for Wireless Innovation (CWI), Queen's University Belfast (QUB). Before joining QUB,  she was an Assistant and then Associate Professor with the Faculty of Engineering, Shahrekord University, Shahrekord, Iran (2015-2021). 
Her research interests include physical-layer security, massive  MIMO, cell-free massive  MIMO, full-duplex communications, and resource management and optimization. She has co-authored many research papers in wireless communications. She has actively served as the reviewer for a variety of IEEE journals,  such as TWC, TCOM, and TVT.
\end{IEEEbiography}

\vspace{-2em}

\begin{IEEEbiography}[{\includegraphics[width=1in,height=1.25in,clip,keepaspectratio]{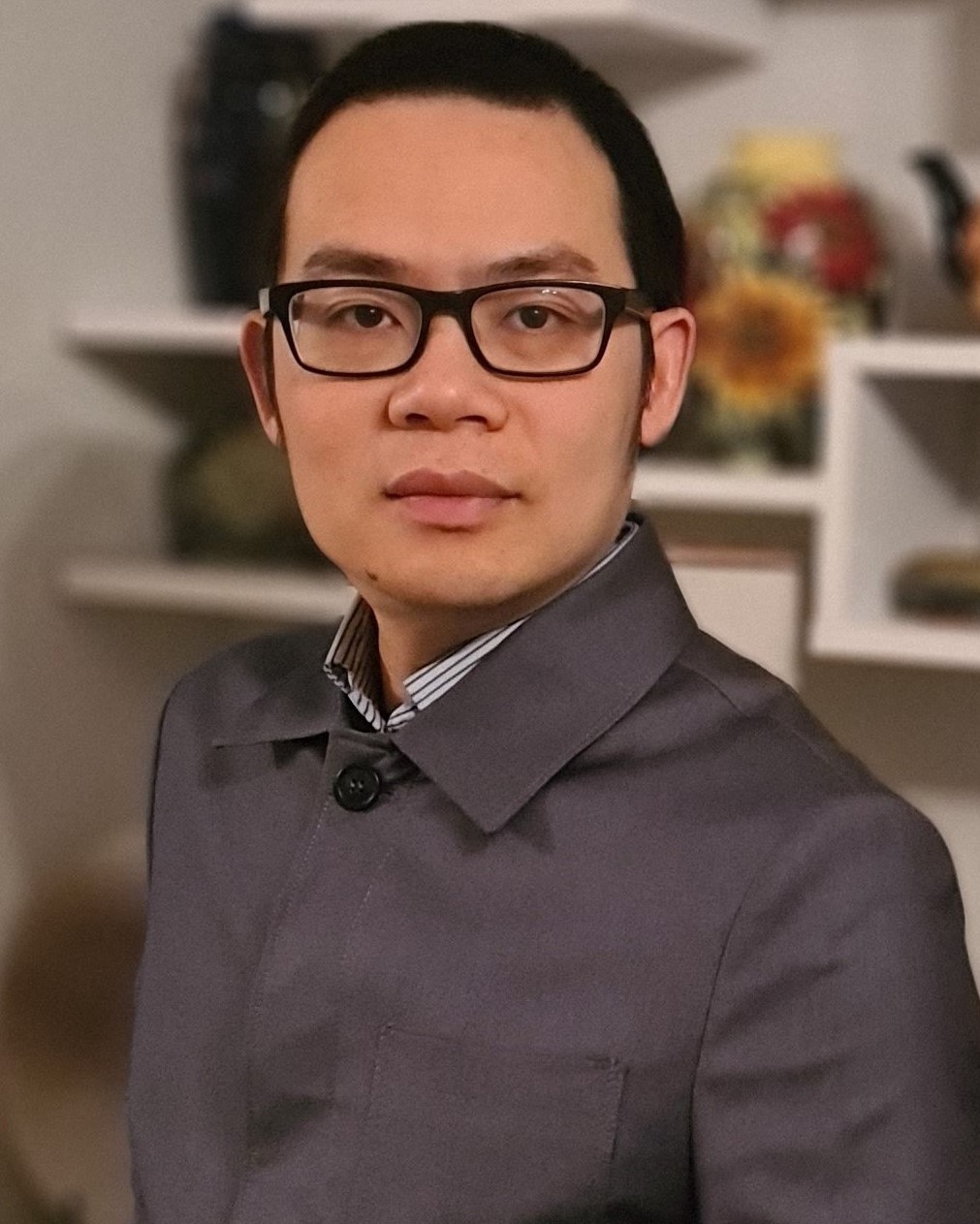}}]
{Hien Quoc Ngo} (Fellow, IEEE)  is currently a Reader with Queen's University Belfast, U.K. His main research interests include massive MIMO systems, cell-free massive MIMO, reconfigurable intelligent surfaces, physical layer security, and cooperative communications. He has co-authored many research papers in wireless communications and co-authored the Cambridge University Press textbook \emph{Fundamentals of Massive MIMO} (2016).

He received the IEEE ComSoc Stephen O. Rice Prize in 2015, the IEEE ComSoc Leonard G. Abraham Prize in 2017, the Best Ph.D. Award from EURASIP in 2018, and the IEEE CTTC Early Achievement Award in 2023. He also received the IEEE Sweden VT-COM-IT Joint Chapter Best Student Journal Paper Award in 2015. He was awarded the UKRI Future Leaders Fellowship in 2019. He serves as the Editor for the IEEE Transactions on Wireless Communications, IEEE Transactions on Communications, the Digital Signal Processing, and the Physical Communication (Elsevier). He was an Editor of the IEEE Wireless Communications Letters, a Guest Editor of IET Communications, and a Guest Editor of IEEE ACCESS in 2017.
\end{IEEEbiography}

\vspace{-2em}
\begin{IEEEbiography}[{\includegraphics[width=1in,height=1.35in,clip,keepaspectratio]{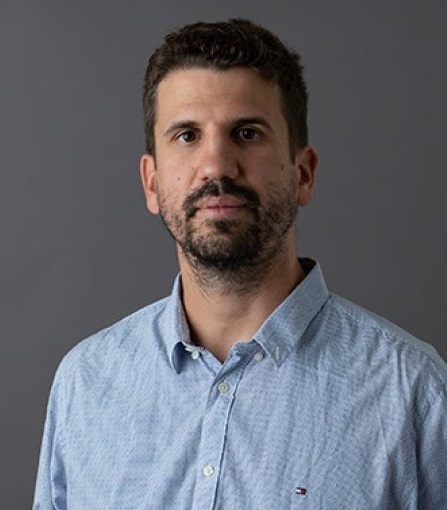}}]
{Michail Matthaiou}(Fellow, IEEE) obtained his Ph.D. degree from the University of Edinburgh, U.K. in 2008. 
He is currently a Professor of Communications Engineering and Signal Processing and Deputy Director of the Centre for Wireless Innovation (CWI) at Queen’s University Belfast, U.K. He is also an Eminent Scholar at the Kyung Hee University, Republic of Korea. He has held research/faculty positions at Munich University of Technology (TUM), Germany and Chalmers University of Technology, Sweden. His research interests span signal processing for wireless communications, beyond massive MIMO, reflecting intelligent surfaces, mm-wave/THz systems and AI-empowered communications.

Dr. Matthaiou and his coauthors received the IEEE Communications Society (ComSoc) Leonard G. Abraham Prize in 2017. He currently holds the ERC Consolidator Grant BEATRICE (2021-2026) focused on the interface between information and electromagnetic theories. To date, he has received the prestigious 2023 Argo Network Innovation Award, the 2019 EURASIP Early Career Award and the 2018/2019 Royal Academy of Engineering/The Leverhulme Trust Senior Research Fellowship. His team was also the Grand Winner of the 2019 Mobile World Congress Challenge. He was the recipient of the 2011 IEEE ComSoc Best Young Researcher Award for the Europe, Middle East and Africa Region and a co-recipient of the 2006 IEEE Communications Chapter Project Prize for the best M.Sc. dissertation in the area of communications. He has co-authored papers that received best paper awards at the 2018 IEEE WCSP and 2014 IEEE ICC. In 2014, he received the Research Fund for International Young Scientists from the National Natural Science Foundation of China. He is currently the Editor-in-Chief of Elsevier Physical Communication, a Senior Editor for \textsc{IEEE Wireless Communications Letters} and \textsc{IEEE Signal Processing Magazine}, an Area Editor for \textsc{IEEE Transactions on Communications} and Editor-in-Large for \textsc{IEEE Open Journal of the Communications Society}. 
%He is an IEEE and AAIA Fellow.
\end{IEEEbiography}

\end{document}